\documentclass[onecolumn]{IEEEtran}
\usepackage{amsthm}
\usepackage{times,amssymb,amsmath,amsfonts,float,nicefrac,color,bbm,mathrsfs,caption,float}
\usepackage{enumerate,multirow,caption,tikz,graphicx,enumitem,soul}
\usepackage[mathscr]{eucal}
\usepackage{sidecap}
\usepackage{algpseudocode}
\usepackage{verbatim}
\usepackage{epstopdf}
\usepackage{textcomp}
\usetikzlibrary{shapes,arrows}
\usepackage{mathabx}
\usepackage[noadjust]{cite}
\usepackage{booktabs}
\usepackage[top=1in,bottom=1in,left=1.2in,right=1.2in]{geometry}
\allowdisplaybreaks

\usepackage{hyperref} 
\usepackage[font={it}]{caption}
\usepackage[margin=0.05in]{subcaption}
\usepackage[bottom]{footmisc}
\usepackage[normalem]{ulem}
\usepackage[ruled,vlined,linesnumbered]{algorithm2e}

\newcommand\nc\newcommand
\nc\bfa{{\boldsymbol a}}\nc\bfA{{\boldsymbol A}}\nc\cA{{\mathscr A}}
\nc\bfb{{\boldsymbol b}}\nc\bfB{{\boldsymbol B}}\nc\cB{{\mathscr B}}
\nc\bfc{{\boldsymbol c}}\nc\bfC{{\boldsymbol C}}\nc\cC{{\mathscr C}}
\nc\bfd{{\boldsymbol d}}\nc\bfD{{\boldsymbol D}}\nc\cD{{\mathscr D}}
\nc\bfe{{\boldsymbol e}}\nc\bfE{{\boldsymbol E}}\nc\cE{{\mathscr E}}
\nc\bff{{\boldsymbol f}}\nc\bfF{{\boldsymbol F}}\nc\cF{{\mathscr F}}
\nc\bfg{{\boldsymbol g}}\nc\bfG{{\boldsymbol G}}\nc\cG{{\mathscr G}}
\nc\bfh{{\boldsymbol h}}\nc\bfH{{\boldsymbol H}}\nc\cH{{\mathscr H}}
\nc\bfi{{\boldsymbol i}}\nc\bfI{{\boldsymbol I}}\nc\cI{{\mathcal I}}
\nc\bfj{{\boldsymbol j}}\nc\bfJ{{\boldsymbol J}}\nc\cJ{{\mathscr J}}
\nc\bfk{{\boldsymbol k}}\nc\bfK{{\boldsymbol K}}\nc\cK{{\mathscr K}}
\nc\bfl{{\boldsymbol l}}\nc\bfL{{\boldsymbol L}}\nc\cL{{\mathscr L}}
\nc\bfm{{\boldsymbol m}}\nc\bfM{{\boldsymbol M}}\nc{\cM}{{\mathscr M}}
\nc\bfn{{\boldsymbol n}}\nc\bfN{{\boldsymbol N}}\nc\cN{{\mathscr N}}
\nc\bfo{{\boldsymbol o}}\nc\bfO{{\boldsymbol O}}\nc\cO{{\mathscr O}}
\nc\bfp{{\boldsymbol p}}\nc\bfP{{\boldsymbol P}}\nc\cP{{\mathscr P}}\nc\eP{{\EuScriptP}}\nc\fP{{\mathfrak P}}
\nc\bfq{{\boldsymbol q}}\nc\bfQ{{\boldsymbol Q}}\nc\cQ{{\mathscr Q}}
\nc\bfr{{\boldsymbol r}}\nc\bfR{{\boldsymbol R}}\nc\cR{{\mathscr R}}
\nc\bfs{{\boldsymbol s}}\nc\bfS{{\boldsymbol S}}\nc\cS{{\mathscr S}}
\nc\bft{{\boldsymbol t}}\nc\bfT{{\boldsymbol T}}\nc\cT{{\mathscr T}}
\nc\bfu{{\boldsymbol u}}\nc\bfU{{\boldsymbol U}}\nc\cU{{\mathscr U}}
\nc\bfv{{\boldsymbol v}}\nc\bfV{{\boldsymbol V}}\nc\cV{{\mathscr V}}
\nc\bfw{{\boldsymbol w}}\nc\bfW{{\boldsymbol W}}\nc\cW{{\mathscr W}}
\nc\bfx{{\boldsymbol x}}\nc\bfX{{\boldsymbol X}}\nc\cX{{\mathscr X}}
\nc\bfy{{\boldsymbol y}}\nc\bfY{{\boldsymbol Y}}\nc\cY{{\mathscr Y}}
\nc\bfz{{\boldsymbol z}}\nc\bfZ{{\boldsymbol Z}}\nc\cZ{{\mathscr Z}}

\newtheorem{theorem}{Theorem}
\newtheorem{lemma}[theorem]{Lemma}

\newtheorem{proposition}[theorem]{Proposition}

\newtheorem{definition}{Definition}

\theoremstyle{remark}
\newtheorem{remark}{Remark}

\nc{\blue}[1]{\textcolor{blue}{#1}}

\DeclareMathOperator{\rank}{rank}

\DeclareMathOperator{\trace}{tr}

\DeclareMathOperator{\Span}{Span}
\newcommand{\ff}{{\mathbb F}}

\begin{document}
		\title{Enabling optimal access and error correction for the repair of Reed-Solomon codes}
		\author{\IEEEauthorblockN{Zitan Chen} \hspace*{1in}
\and \IEEEauthorblockN{Min Ye} \hspace*{1in}
\and \IEEEauthorblockN{Alexander Barg}}
		\date{}
		\maketitle
{\renewcommand{\thefootnote}{}\footnotetext{
\vspace{.1in}
\hspace*{-.3in}\rule{1.5in}{.4pt}

Zitan Chen is with Dept. of ECE and ISR, University of Maryland, College Park, MD 20742. Email: chenztan@umd.edu. His research was supported by NSF grants CCF1618603 and CCF1814487.

Min Ye,  Email: yeemmi@gmail.com.

Alexander Barg is with Dept. of ECE and ISR, University of Maryland, College Park, MD 20742 and also with IITP, Russian Academy of Sciences, 127051 Moscow, Russia. Email: abarg@umd.edu. His research was supported by NSF grant CCF1618603 and CCF1814487.
}}
\renewcommand{\thefootnote}{\arabic{footnote}}
\setcounter{footnote}{0}

%{\small{\tableofcontents}}

\vspace{-.2in}
 \begin{abstract} Recently Reed-Solomon (RS) codes were shown to possess a repair scheme that supports repair of failed nodes with optimal repair bandwidth.  In this paper, we extend this result in two directions. First, we propose a new repair scheme for the RS codes constructed in [Tamo-Ye-Barg, {\em IEEE Transactions on Information Theory}, vol. 65, May 2019] and show that our new scheme is robust to erroneous information provided by the helper nodes while maintaining the optimal repair bandwidth. Second, we construct a new family of RS codes with optimal access for the repair 
of any single failed node. We also show that the constructed codes can accommodate both features, supporting optimal-access repair with 
optimal error-correction capability.

Going beyond RS codes, we also prove that any scalar MDS code with optimal repair bandwidth allows for a repair scheme with optimal access property. 
 \end{abstract}

\section{Introduction}\label{sec:introduction}

The problem of efficient erasure correction in various classes of algebraic codes has recently attracted renewed attention because of its links to 
applications of erasure coding for distributed storage.	Compared to the classic setting of erasure correction, efficient functioning of distributed storage systems critically depends on the volume of communication exchanged between the nodes for the purposes of data recovery. The constraint on the amount of communication, termed ``repair bandwidth,'' adds new features to the problem, and has motivated a large amount of research
in coding theory in the last decade.

Consider an $(n,k,l)$ array code $\cC$ over a finite field $F$, i.e.,  a 
collection of codewords $c=(c_1,\dots,c_n)$, where $c_i=(c_{i,0},c_{i,1},\dots,c_{i,l-1})^T\in F^l, i=1,\dots, n$. A node $c_i,i\in [n]$ 
can be repaired from a subset of $d\ge k$ helper nodes $\{c_j:j\in\cR\},\cR\subseteq[n]\backslash \{i\},$ by downloading $\beta_i(\cR)$ symbols of 
$F$ if there are
numbers $\beta_{ij}, j\in\cR$,
functions $f_{ij}: F^l\to F^{\beta_{ij}}, j\in\cR,$
and a function $g_i: F^{\sum_{j\in\cR}\beta_{ij}}\to F^l$
such that
   $$
   c_i=g_i(\{f_{ij}(c_{j}),j\in\cR \})
\text{~for all~} c=(c_1,\dots,c_n)\in \cC
   $$
   and
   $$
\sum_{j\in\cR}\beta_{ij}=\beta_i(\cR).
   $$
Codes that we consider form linear spaces over $F$. If $\cC$ is not linear over $F^l$, it is also called a {\em vector} code, while if it is, it is called {\em scalar} to stress the linearity property. A code $\cC$ is called MDS if any $k$ coordinates $\{c_{j_i},i=1,\dots,k\}$ of the codeword suffice to recover its remaining $n-k$ coordinates. In this paper we study the repair problem of scalar MDS codes.

It is well known \cite{dimakis2010network} that for any MDS code $\cC$ (scalar or vector), any $i\in[n]$, and any $\cR\subseteq [n]\setminus\{i\}$ of cardinality $|\cR|\ge k$, we have
\begin{equation} \label{eq:betai}
\beta_i(\cR)\ge \frac{l}{|\cR|-k+1}|\cR|.
\end{equation}
For an MDS code $\cC$, we define the minimum bandwidth of repair of a node from a $d$-subset $\cR$ of helper nodes as
$\beta(d)=\max_{i\in[n]}\min_{\cR\subset[n]\backslash\{i\}, |\cR|=d}\beta_i(\cR).$  It follows immediately from \eqref{eq:betai} that
   \begin{equation}\label{eq:bandwidth}
   \beta:=\beta(d)\ge \frac{l}{d-k+1}d.
\end{equation}
%\vspace*{-.1in}
An MDS code that attains the bound \eqref{eq:bandwidth} with equality is said to afford optimal repair, and a repair scheme that attains this bound 
is called optimal.  Such codes are also termed {\em minimium storage regenerating} or MSR codes, and the parameter $l$ is called node size or 
sub-packetization. Multiple constructions of vector MDS codes with optimal repair are available in the literature, including papers \cite{Rashmi11}, 
\cite{Tamo13}, \cite{Ye16,Ye16a}, \cite{Goparaju17}, \cite{Raviv17}.

The basic repair problem of MDS codes 
has been extended to the case that some of the helper nodes provide erroneous information (or arbitrary nature). Suppose that a subset of $e$ nodes
out of $d$ helpers provide erroneous information and define $\beta(d,e)$ to be the minimum number of symbols needed to repair a failed node
in the presence of such errors. It was shown 
  \cite{Rashmi12}, \cite{Pawar11} that for $d\ge k+2e$,
  \begin{equation}\label{eq:bandwidth-errors}
\beta(d,e)\ge \frac{dl}{d-2e-k+1}.
  \end{equation}
A repair scheme that achieves this bound is said to have {\em optimal error correction capability}.
Constructions of MDS array codes with optimal error correction capability are presented, for instance, in \cite{Ye16}. 

Another parameter of erasure codes for distributed storage that affects the system performance is the so-called access, or input-output cost
of repair. Indeed, while the code may support parsimonious exchange between the helper nodes and the repair center, generation of
the symbols to be transmitted from the helper node may require reading the entire contents of the node (trivial access), which increases delays 
in the system. The smallest number of symbols accessed on each of the helper nodes in an MSR code is $l/(d-k+1),$ and such codes
are said to have the {\em optimal access property}. Advantages of having this property are well recognized in the literature starting with
\cite{shah2012distributed}, and a number of papers were devoted to constraints that it imposes on the code parameters such as 
sub-packetization \cite{Tamo14}, \cite{balaji2017tight}. Many families of MSR codes including early constructions in \cite{Cadambe11,Tamo13} as well as code families for general parameters in \cite{Ye16,Ye16a}, \cite{Clay18}, \cite{chen2019explicit} have the optimal access property.

The optimal-access repair and optimal error correction capability can be combined. According to \eqref{eq:bandwidth-errors}, we say that a code family/repair scheme have both properties if repair can be performed in the presence of $e$ errors, while the number of symbols accessed on each of the helper nodes equals $1/(d-2e-k+1)$ proportion of the contents of each of $d$ helper nodes.

While the aforementioned papers mostly deal with vector codes, in this paper we focus on the repair problem for scalar MDS codes, more specifically, for Reed-Solomon (RS) codes. This code family continues to attract attention in 
multiple aspects of theoretical research (list decoding of variants of RS codes, locally recoverable codes, to name a few) and it is also one of the 
most used coding methods in a vast variety of practical systems. The first work to isolate and advance the repair problem for RS codes was
\cite{Guruswami16} which itself followed and developed the ideas in \cite{Shanmugam14}. In \cite{Guruswami16}, the authors view each coordinate of 
RS codes as a vector over some subfield and characterize linear repair schemes of RS codes over this subfield. For RS codes (and more generally for 
scalar codes), the node size $l$ is defined as the degree of extension of the symbol field over the subfield. Following \cite{Guruswami16}, several 
papers attempted to optimize the repair bandwidth of RS codes \cite{Dau16}, \cite{Dau17}, \cite{Mardia19}. A family of optimal-repair RS codes in 
the case of repairing a single failed node as well as multiple nodes was constructed in \cite{Tamo18}. Later this construction was extended to the 
case of the rack-aware storage model, resulting in a family of codes with optimal repair of a single node \cite{chen2019explicit}, and this problem 
was addressed again in \cite{JinLuoXing19}.

In this paper we address two problems related to RS repair, namely, 
\begin{itemize}\item[(i)] repair schemes of RS codes with optimal error correction, and \item[(ii)]
RS codes with optimal-access repair.
\end{itemize}
Error correction during repair of failed storage nodes was previously only considered for vector codes \cite{Rashmi12}, \cite{Pawar11}, \cite{Ye16}.
The problem of low-access RS codes was studied in \cite{DauViterbo18,DauDuursmaChu18,LiDauWang19}. In particular, the last of these works analyzed the access (input/output) cost of the family of RS codes of \cite{Tamo18}, providing an estimate of this
parameter, but stopping short of achieving optimal access.

Our main results provide a solution to problems (i)-(ii). 
Specifically, we construct a repair scheme for RS codes in \cite{Tamo18} that has optimal error
correction capability (i.e., attains the bound \eqref{eq:bandwidth-errors}), and we also construct a family of RS codes with optimal access repair for any single failed node. Additionaly, we prove that the constructed codes can be furnished with a repair scheme that supports both optimal error correction and 
optimal-access repair.

Apart from this, we also show that any scalar MDS code with optimal repair of a single node from $d$ helpers, $k\le d\le n-1,$ affords a repair 
scheme with optimal access, and this includes the RS codes in \cite{Tamo18}. While our arguments do not provide an explicit construction, 
we give a combinatorial search procedure, showing that it exists for any scalar MSR code. The resulting optimal access codes have the
same sub-packetization as the original MDS codes.

The constructions are technically involved, and we begin in Sec.~\ref{sec:example} with illustrating them in an example.
The three sections that follow it are devoted to the results described above.

%%%%%%%%%%%%%%%%%%%%%%%%
%
% Section EXAMPLE: REPAIR OF 1ST NODE
%
%%%%%%%%%%%%%%%%%%%%%%%%

\section{A simple example}\label{sec:example}
In this section, we construct an RS code together with a repair scheme that can recover its {\em first node} with both optimal access and optimal error correction capability. 

\subsection{Preliminaries}
\subsubsection{}We begin with some standard definitions. Recall that a {\em generalized RS code} (GRS code) of length $n$ and dimension $k$ over a finite field $F$
is obtained by fixing a set of $n$ distinct evaluation points $\Omega:=\{\alpha_1,\alpha_2,\dots,\alpha_n\}\subset F$ and a vector 
$(v_1,\dots, v_n)\in (F^\ast)^n$ with no zero coordinates. Then the GRS code is the set of vectors
   $$
   {\text {GRS}}_{F}(n,k,v,\Omega)=\{(v_1f(\alpha_1),v_2f(\alpha_2),\dots,v_nf(\alpha_n)): f\in F[x], \deg f < k \}.
   	$$
In particular, if $(v_1,\dots, v_n)=(1,\dots,1),$ then the GRS code is called the Reed-Solomon (RS) code and is denoted by $\text{RS}_{F}(n,k,
\Omega).$ It is a classic fact that the dual code $(\text{RS}_{F}(n,k,\Omega))^\bot$ is ${\text {GRS}}_{F}(n,n-k,v,\Omega),$ where $v\in 
(F^{\ast})^n$ is some vector. In particular, if $c=(c_1,\dots,c_n)\in F^n$ is a vector such that $\sum_{i=1}^n c_ih(\alpha_i)=0$ for every polynomial $h(x)$ of degree $\le k-1,$ then $c$ is contained in a ${\text {GRS}}_{F}$ code of dimension $n-k.$ Rephrasing this, we have
the following obvious proposition that will be frequently used below.
\begin{proposition}\label{prop:RS} Let $c=(c_1,\dots,c_n)\in F^n$ and suppose that $\sum_{i=1}^n c_i\alpha_i^t=0$ for all $t=0,1,\dots,k-1.$
Then the vector $c$ is contained in a code $\text{GRS}_F(n,n-k,v,\Omega)$, where $v\in (F^\ast)^n$ and $\Omega=\{\alpha_1,\dots,\alpha_n\}.$
%GRS code of length $n$ and dimension $n-k.$
\end{proposition}

Let $E$ be an algebraic extension of $F$ of degree $s$. The {\em trace mapping} $\trace_{E/F}$ is given by
$x\mapsto 1+x^{|F|}+x^{|F|^2}+\dots+x^{|F|^{s-1}}.$ For any basis $\gamma_0,\dots,\gamma_{s-1}$ of $E$ over $F$ there
exists a {\em trace-dual basis} $\delta_0,\dots,\delta_{s-1}$, which satisfies $\trace_{E/F}(\gamma_i \delta_j)={\mathbbm 1}_{\{i=j\}}$ for all pairs $i,j.$ For an element $x\in E$ the coefficients of its expansion in the basis $(\gamma_i)$ are found using the dual basis, specifically,
$x=\sum_{i=0}^{s-1} \trace_{E/F}(x\delta_i)\gamma_i.$ As a consequence,  for any basis $(\delta_i)$ the mapping $E\to F^s$ given by $x\mapsto (\trace(x\delta_i),i=0,\dots,s-1)$ is a bijection.

\vspace*{.1in}\subsubsection{}Before we define the RS code that will be considered below, let us fix the parameters of the repair scheme. We attempt to repair a failed node using information from $d$ helper nodes. Suppose that at most $e$ of them provide erroneous information.
Assume that $d-2e\ge k$, and let $s:=d-2e-k+1$. Let $F$ be a finite field of size $|F|\ge n-1.$
Choose a set of distinct evaluation points $\Omega:=\{\alpha_1,\alpha_2,\dots,\alpha_n\}$ such that $\alpha_i\in F$ for all $2\le i\le n$ and $\alpha_1$ is an algebraic element of degree $s$ over $F$. Let $E:=F(\alpha_1).$ Consider the code 
   $$\cC:=\text{RS}_{E}(n,k,\Omega).$$
    In this section we present a repair scheme of the code $\cC$ that can repair the first node of $\cC$ over the field $F$; in other words, we represent the coordinates of $\cC$ as $s$-dimensional vectors over $F$ in some basis of $E$ over $F.$ Thus, the node size of this code is $s$. We note that the code $\cC$ represented in this way is still a scalar code. 

The repair scheme presented below has the following two properties:

\begin{itemize}
\item the optimal error correction capability, i.e., the repair bandwidth achieves the bound \eqref{eq:bandwidth-errors} for any pair $(d,e)$ such that $d-2e=s+k-1$;
\item in the absence of errors it has the optimal access property, i.e., the number of symbols accessed during the repair process is $d$. 
Thus, in this case $e=0$ and $s=d-k+1$.
\end{itemize}

\subsection{Repair scheme with optimal error correction capability}
\label{sect:toycorrection}
Let $c=(c_1,c_2,\dots,c_n)\in \cC$ be a codeword and suppose that $c_1$ is erased.
Since $\cC^\bot={\text {GRS}}_E(n,n-k,v,\Omega)$ for some $v\in (E^\ast)^n,$ we have
$$
v_1 \alpha_1^t c_1 + v_2 \alpha_2^t c_2 + \dots + v_n \alpha_n^t c_n=0, \quad \quad
t=0,1,\dots,n-k-1,
$$
or
\begin{equation} \label{eq:bbq}
v_1 \alpha_1^t c_1 = - v_2 \alpha_2^t c_2 - \dots - v_n \alpha_n^t c_n,
\quad \quad t=0,1,\dots,n-k-1 .
\end{equation}
Evaluating the trace $\trace=\trace_{E/F}$ on both sides of \eqref{eq:bbq}, we obtain the relation
\begin{align}
\trace(v_1 \alpha_1^t c_1) & = - \trace(v_2 \alpha_2^t c_2) - \dots - \trace(v_n \alpha_n^t c_n) \nonumber\\
& = -\alpha_2^t \trace(v_2 c_2) -\dots 
-\alpha_n^t \trace(v_n c_n), 
\quad \quad t=0,1,\dots,n-k-1 ,\label{eq:n-k}
\end{align}
where the second equality follows from the fact that $\alpha_2,\dots,\alpha_n\in F$.
Therefore, knowing the values of $(\trace(v_2 c_2),\dots,\trace(v_n c_n))$ enables us to compute $\trace(v_1 \alpha_1^t c_1)$ for all $0\le t \le n-k-1$.
Since $\deg_F(\alpha_1)=s$, the elements $1,\alpha_1,\dots,\alpha_1^{s-1}$ form a basis of $E$ over $F$.
As a consequence, one can recover $c_1$ from the values of $\{\trace(v_1 \alpha_1^t c_1) : 0\le t \le s-1 \}$.
By definition, $s-1=d-2e-k\le  n-k-1$, so $\{\trace(v_1 \alpha_1^t c_1) : 0\le t \le s-1 \}\subseteq \{\trace(v_1 \alpha_1^t c_1) : 0\le t \le n-k-1 
\}$. Combining this with \eqref{eq:n-k}, we see that the value $c_1$ is fully determined by the set of elements $(\trace(v_2 c_2),\dots,\trace(v_n 
c_n))$.

Recalling our problem, we will show that in order to repair $c_1$, it suffices to acquire the values $\trace(v_i c_i)$ from any $d$ helper nodes 
provided that at least $d-e=(d+s+k-1)/2$ of these values are correct. This will follow from the following proposition.
\begin{proposition}\label{prop:example-dual} {Let $f(x)\in F[x]$ be the minimal polynomial of $\alpha_1$.} For any $s< n-k$ 
and any $c=(c_1,\dots,c_n)\in \cC$ the vectors {$(f(\alpha_2)\trace(v_2 c_2),\dots,f(\alpha_n)\trace(v_n c_n))$}
are contained in an $(n-1,s+k-1)$ GRS code over $F.$ % of length $n-1$ and dimension $s+k-1$.
\end{proposition}
\begin{proof}
Let $\cT:=\{0,1,\dots,n-k-s-1\}.$ Since $\alpha_i\in F,i=2,\dots,n$ by definition we have $f(\alpha_i)\ne0$ for all such $i$.
Next, $\deg(f)=s,$ and thus for all $t\in \cT$
   $$
   (v_1 \alpha_1^t f(\alpha_1), v_2 
\alpha_2^t f(\alpha_2), \dots, v_n \alpha_n^t f(\alpha_n)) \in \cC^\bot.
$$
This implies that for all $t\in \cT$
  $$
  v_1 \alpha_1^t f(\alpha_1)c_1 + v_2 \alpha_2^t f(\alpha_2)c_2+ \dots+ v_n \alpha_n^t f(\alpha_n)c_n=0, \quad
  $$
but $f(\alpha_1)=0,$ so taking the trace, we obtain
  \begin{equation}\label{eq:example-parity}
 \alpha_2^t f(\alpha_2) \trace(v_2 c_2)+ \dots+ \alpha_n^t f(\alpha_n) \trace(v_n c_n)=0, \quad
t\in \cT.
  \end{equation}
By Proposition \ref{prop:RS}, this implies that the vectors $(f(\alpha_2)\trace(v_2 c_2),\dots,f(\alpha_n)\trace(v_n c_n))$ are contained in a GRS code of length $n-1$ with $n-s-k$ parities.
\end{proof}

The GRS code identified in this proposition can be punctured to any subset $\cR$ of $d$ coordinates, retaining the dimension and the MDS property. This means that the punctured code is capable of correcting any $e=(d-s-k+1)/2$ errors. Therefore, as long as no more than $e$ helper nodes 
provide incorrect information, we can always recover $(\trace(v_2 c_2),\dots,\trace(v_n c_n))$ by acquiring a subset $\{\trace(v_{i_j} c_{i_j}), 
j=1,\dots,d\}$ from any $d$ helper nodes and correcting the errors based on any decoding procedure of the underlying MDS code.
Finally note that the case $s=n-k$ can be added trivially because then $d=n-1$ and $e=0,$ so all the helper nodes provide accurate information, and no error correction is required (or possible).

\subsection{Optimal access property} \label{sect:toyaccess}
Following the discussion in the first part of this section, we show that the code $\cC=\text{RS}_E(n,k,\Omega)$ defined above supports optimal-access repair of 
the node $c_1.$ In this part we assume that the helper nodes provide accurate information about their contents, and we do not attempt
error correction. 

To represent the code, we choose a pair of trace-dual bases $(b_i), (b^\ast_i)$ of $E$ over $F,$ where we assume w.l.o.g. that
$b_0=1.$ Next, represent the $i$th coordinate of the code, $i\in\{1,\dots,n\},$ 
using the basis $(v_i^{-i}b_m, m=0,\dots,s-1),$ where $(v_1,\dots,v_n)$ is defined by the code $\cC^\bot.$ Namely, for a codeword $c\in \cC$ we have 
   \begin{equation}  \label{eq:expansion}
c_i=v_i^{-1}\sum_{m=0}^{s-1} c_{i,m}b_m^\ast ,
\end{equation}
where $c_{i,m}\in F$ for all $m=0,1,\dots, s-1$. We assume that each storage node contains the vector $(c_{i,0}, c_{i,1}, \dots,c_{i,s-1}).$

As discussed above, the value $c_1$ can be recovered from any $d$-subset of the set of elements $\{\trace(v_{i} c_{i}), j=2,\dots, n\}.$ Further, 
for all $i=2,\dots,n$ and $m=0,\dots, s-1$ we have $\trace(v_ic_i b_m)=c_{i,m},$ so in particular,
   $$
   \trace(v_ic_i)=c_{i,0}.
  $$
Thus, to repair $c_1$ it suffices to access and download a single symbol $c_{i,0}$ from the chosen subset of $d$ helper nodes. According to the bound \eqref{eq:betai}, the minimum number of symbols downloaded from a helper node for optimal repair is the $(1/s)$th proportion of the node's contents. Overall this shows that the repair scheme considered above has the optimal access property.

The above discussion sets the stage for constructing RS codes with optimal-access repair for each of the $n$ coordinates. 
Namely, we took a basis $1,b_1,\dots,b_{s-1}$ of $E$ over $F$ and represented each $c_i$ in the  basis $(v_i^{-1}
b^\ast_i).$ The only element of the helper coordinate that we access and download is $c_{i,0}.$
For more complicated constructions of RS codes, e.g., the ones constructed in \cite{Tamo18} and below in the paper, we assume that $E$ is
an $l$-degree extension of $F$. The known repair schemes require to 
download elements of the form $\trace(v_i c_i a_0), \trace(v_i c_i a_1), \dots, \trace(v_i c_i a_{(l/s)-1})$, where $a_0, a_1, \dots, a_{(l/s)-1}$ are 
linearly independent over $F$. In this case, we can extend the set $a_0, a_1, \dots, a_{(l/s)-1}$ to a basis $(b_i)$ of $E$ over $F.$
Following the approach in \eqref{eq:expansion}, we store the code coordinate $c_i$ as the vector of its coefficients $(c_{i,0}, c_{i,1}, \dots,c_{i,l-1})$ in the
dual basis $(b_i^\ast, i=0,\dots,l-1)$ of the basis $(b_i).$ 
Since $c_{i,m}=\trace(v_i c_i a_m)$ for all $m=0,1,\dots,l/s-1$, this choice of the basis enables one to achieve optimal access. This idea underlies the construction presented below in Sec.~\ref{sect:new}.

\subsection{Optimal access with error correction}
Thus far, we have assumed that errors are absent for optimal-access repair. To complete the picture, we address the case of codes with both optimal access and optimal error correcting capability for the repair of node $c_1.$ It is easily seen that both properties can be combined. 
Indeed, since $\trace(v_ic_i)=c_{i,0}$ for all $i=2,\dots, n$, and since by Proposition~\ref{prop:example-dual} these elements 
form a codeword of a GRS code, it is immediately clear that $c_1$ can be repaired with optimal error correction capability and optimal access. 
To enable this property for any $c_i$, below we add extra features to the general repair scheme with optimal access. Specifically, error correction and optimal access are based on two different structures supported by the code. We show that it is possible to realize the error-correction structure in an extension field located between the base field and the symbol field of the code. Further reduction to the base field enables us
to perform repair with optimal access. These ideas are implemented in detail in Sec.~\ref{sect:err-oa} below.

%%%%%%%%%%%%%%%%%%%%%%%%
%
% Section ERROR CORRECTION
%
%%%%%%%%%%%%%%%%%%%%%%%%

\section{Enabling error correction for repair of RS codes of \cite{Tamo18}}
\label{sect:erc}

In this section we propose a new repair scheme for the optimal-repair family of RS codes of \cite{Tamo18} that supports the optimal error correction 
capability.

\subsection{Preliminaries} We begin with briefly recalling the definition of the subfamily of RS codes of \cite{Tamo18}. The construction 
depends on the number of helper nodes $d$ used for the purpose of repair of a single node, $k\le d\le n-1.$
\begin{definition}[\cite{Tamo18}]\label{def:RSerrors} Let $p$ be a prime, let $s:=d-k+1,$ and let $p_1,\dots,p_n$ be distinct primes that satisfy the condition $p_i\equiv 1\;\text{mod}\, s, 
i=1,\dots,n,$ Let $\cC:=RS_{\mathbb K}(n,k,\Omega)$ be a Reed-Solomon code, where
\begin{itemize}
\item $\Omega=\{\alpha_1,\dots,\alpha_n\},$ where $\alpha_i, i=1,\dots,n$ is an algebraic element of degree $p_i$ over $\ff_p,$
\item $\mathbb{K}=\ff(\beta),$ where $\beta$ is an algebraic element of degree $s$ over $\ff:=\ff_p(\alpha_1,\dots,\alpha_n).$
\end{itemize}
\end{definition}
%Let $\ff_p$ be a field of prime order $p$ and let $p_1,\dots,p_n$ be distinct primes that satisfy the condition $p_i\equiv 1\;\text{mod}\, s, 
%i=1,\dots,n,$ where in the original construction $s=d-k+1$. Consider the set $\Omega=\{\alpha_1,\dots,\alpha_n\},$ where $\alpha_i, i=1,\dots,n$ is an algebraic element of degree $p_i$ over $\ff_p.$ Let $\ff=\ff_p(\alpha_1,\dots,\alpha_n)$ and let $\mathbb{K}=\ff(\beta),$ where $\beta$ is an element of degree $s$ over $\ff$. Finally, consider the RS code $\cC:=RS_{\mathbb K}(n,k,\Omega).$ 

As shown in \cite{Tamo18}, this code supports optimal repair of any node $i$ from any set of $d$ helper nodes in $[n]\backslash\{i\}.$ Below we use this construction, choosing
the value of $s$ based not only on the number of helpers but also on the target number of errors tolerated by the repair procedure.

In this section we consider an RS code $\cC$ given by Def.~\ref{def:RSerrors}, where we take $s=d-2e-k+1.$ 
%Next we present our new repair scheme with optimal error correction capability for correcting up to $e$ errors.
%Let $d$ be the number of helper nodes, $k\le d\le n-1,$ and let $e$ be the target number of errors. 
%We take $s=d-2e-k+1$ in Def.~\ref{def:RSerrors}. 
For this code we will present a new repair scheme that has the property of optimal error correction.
This repair scheme as well as the original repair scheme developed in \cite{Tamo18} rely on the following lemma:
	\begin{lemma}[\cite{Tamo18}, Lemma 1]
		\label{le:subspace}
Let $F$ be a finite field. Let $r$ be a prime such that $r\equiv 1\;\text{mod}\,s$ for some $s\ge 1.$
Let $\alpha$ be an element of degree $r$ over $F$  and $\beta$ be of degree $s$ over the field $F(\alpha)$. 
Let $K=F(\alpha,\beta)$ be the extension field of degree $rs.$  Consider the $F$-linear subspace $S$ of dimension $r$ with the basis
		\begin{align*}
			E:=\{\beta^u\alpha^{u+qs}\mid u=0,\ldots,s-1;q=0,\ldots,{\textstyle\frac{r-1}s-1 \}} \bigcup \Big\{\sum_{u=0}^{s-1}\beta^u\alpha^{r-1}\Big\}.
		\end{align*}
Then  $S+S\alpha+\cdots+S\alpha^{s-1}=K,$ and this is a direct sum.
	\end{lemma}

Without loss of generality, we only present the repair scheme for the first node $c_1$, and all the other nodes can be repaired in 
the same way (this is different from the previous section where the code was designed to support optimal repair {\em only} of the node $c_1$).
The scheme is complicated, and we take time to develop it, occasionally repeating similar arguments more than once rather than compressing the presentation.

The repair of $c_1$ is conducted over the field $F_1:=\ff_p(\alpha_2,\alpha_3,\dots,\alpha_n).$ It is clear that 
$\ff=F_1(\alpha_1)$ and $\mathbb{K}=\ff(\beta)$, where $\deg_{F_1}(\alpha_1)=p_1$ and $\deg_{\ff}(\beta)=s$. Below we use $
\trace=\mathrm{tr}_{\mathbb{K}/F_1}$ to denote the trace mapping from $\mathbb{K}$ to $F_1$.

Define the set
\begin{equation} \label{eq:defE1}
E_1:=\{\beta^u\alpha_1^{u+qs}\mid u=0,\ldots,s-1;q=0,\ldots,{\textstyle\frac{p_1-1}s-1 \}} \bigcup \Big\{\sum_{u=0}^{s-1}\beta^u\alpha_1^{p_1-1}\Big\}.
\end{equation}
Clearly, $|E_1|=p_1$, and we write the elements in $E_1$ as $e_0,e_1,\dots,e_{p_1-1}$. Then Lemma~\ref{le:subspace} implies that the set of elements
    \begin{equation} \label{eq:base1}
\{e_i \alpha_1^j: i=0,\dots,p_1-1, j=0,\dots,s-1\}
     \end{equation}
forms a basis of $\mathbb{K}$ over $F_1$.

Let $\cC^\bot={\text {GRS}}_{\mathbb K}(n,n-k,v,\Omega)$ be the dual code. For every codeword $(c_1,\dots,c_n)\in\cC$ we have
    $$
v_1 \alpha_1^t c_1 + v_2 \alpha_2^t c_2 + \dots + v_n \alpha_n^t c_n=0, \quad \quad
t=0,1,\dots,n-k-1 .
    $$
Multiplying by $e_i$ on both sides of the equation and evaluating the trace, we obtain the relation
      \begin{equation} \label{eq:tin}
      \begin{aligned}
\trace(e_i v_1 \alpha_1^t c_1) &= -\sum_{j=2}^n  \trace(e_i v_j \alpha_j^t c_j)\\
&= -\sum_{j=2}^n \alpha_j^t \trace(e_i v_j c_j) , \quad 
t=0,1,\dots,n-k-1 ,
   \end{aligned}
     \end{equation}
where the second equality follows since $\alpha_j\in F_1$ for all $2\le j\le n$. Therefore, the elements
$\{\trace(e_i v_j c_j):2\le j\le n\}$ suffice 
to compute $\{\trace(e_i v_1 \alpha_1^t c_1):0\le t\le n-k-1\}$. Since $s=d-2e-k+1\le d-k+1\le n-k$, we can calculate $\{\trace(e_i v_1 \alpha_1^t 
c_1):0\le t\le s-1\}$ from $\{\trace(e_i v_j c_j):2\le j\le n\}$.
Thus knowing the values of $\{\trace(e_i v_j c_j):2\le j\le n, 0\le i\le p_1-1\}$ suffices to find the set of elements
\begin{equation} \label{eq:base2}
\{\trace(e_i v_1 \alpha_1^t c_1):0\le t\le s-1, 0\le i\le p_1-1\}.
\end{equation}
Since the set \eqref{eq:base1} forms a basis of $\mathbb{K}$ over $F_1$, the set
$\{e_i v_1 \alpha^t: 0\le i\le p_1-1, 0\le t\le s-1\}$ also forms a basis of $\mathbb{K}$ over $F_1$, and therefore we can recover $c_1$ from \eqref{eq:base2}.
In conclusion, to recover $c_1$, it suffices to know the set of elements $\{\trace(e_i v_j c_j):2\le j\le n, 0\le i\le p_1-1\}$.

%%%%%%%%%%%%%%%%%%%%%%%%
%
% Section PLAN OF REPAIR SCHEME
%
%%%%%%%%%%%%%%%%%%%%%%%%

\subsection{The repair scheme}\label{sec:plan}

For $j=2,3,\dots,n$ define the vector $r_j:=(\trace(e_i v_j c_j), i=0,\dots,p_1-1).$ In this section we design invertible linear 
transformations $M_j$ that send these vectors to a set of vectors $z_j$ that support error correction. 
The following proposition underlies our repair scheme.
\begin{proposition} \label{prop:plan} Consider the set of vectors $z_j=(z_{j,0},z_{j,1}, \dots, z_{j,p_1-1}),j=2,3,\dots,n$ defined by
\begin{equation} \label{eq:plan} 
z_j^T=M_j r_j^T,
\end{equation}  
where $M_2,\dots,M_n$ are invertible matrices of order $p_1.$ Suppose that for every $i=0,1,\dots,p_1-1,$ the vector $(z_{2,i},z_{3,i},\dots,z_{n,i})$ is contained in an MDS code of length $n-1$ and dimension $s+k-1$.
Then there is a repair scheme of the code $\cC$ that supports recovery of the node $c_1$ with optimal error correction capability.
\end{proposition}
Note that, by the closing remark in Sec.~\ref{sect:toycorrection}, it suffices to assume that $s<n-k.$
\begin{proof} 
If $(z_2,z_3,\dots,z_n)$ is a codeword in
an MDS array code of length $n-1$ and dimension $s+k-1,$ then the punctured codeword $(z_j:j\in\cR)$ is contained in an MDS array code of length 
$d=|\cR|$ and dimension $s+k-1= d-2e$, and such the code can correct any $e$ errors. 

To repair the failed node $c_1$, we download $p_1$-dimensional vectors $\hat r_j, j\in \cR,$ where $\cR\subset [n]\backslash\{1\}, |\cR|=d$ is a set of $d$ helper nodes. For all but $e$ or fewer values of $j$, we have $\hat r_j=r_j$. The repair scheme consists of the following steps:
\begin{enumerate}
% \item[(i)] Download the vectors $\hat r_j, j\in \cR,$ where $\cR\subset [n]\backslash\{1\}, |\cR|=d$ is a set of helper nodes,
%   $p_1d$ symbols of $F_1$ given by $\{\trace(e_i v_j c_j): j\in\cR, i\in\{0,1,\dots,p_1-1\}\}$ 
 \item[(i)] Find the vectors $\hat z_j^T=M_j \hat r_j^T,j\in \cR,$ 
 \item[(ii)] Find the vectors $z_j,j\in\cR$ using the error correction
 procedures of the underlying MDS codes,
 \item[(iii)] For every $i=0,\dots,p_1-1$ use the $d$-subset $\{z_{j,i},j\in \cR\}$ to recover the codeword $(z_{2,i},z_{3,i},\dots,z_{n,i})$,
 \item[(iv)] Find the vectors $r_j^T=M_j^{-1}z_j^T,j=2,\dots,n-1$ and finally recover $c_1.$
\end{enumerate}
Step (ii) is justified by the fact that, by assumption, at most $e$ of the elements $\hat z_j$ are incorrect. In step (iii) we rely on the 
fact that $d$ symbols of the MDS codeword suffice to recover the remaining $n-1-d$ symbols, and in step (iv) we use invertibility of the matrices
$M_j$ and recover $c_1$ using \eqref{eq:tin}, \eqref{eq:base2}.

The total number of downloaded symbols of $F_1$ equals $p_1d,$ and it is easy to verify that the repair bandwidth of our scheme meets the bound
\eqref{eq:bandwidth-errors} with equality.
\end{proof}

Why do we need the matrices $M_j$ and why were they not involved in the example in Sec.~\ref{sect:toycorrection}? The answer is
related to the fact that we need to remove the failed node from consideration and obtain a codeword of the MDS code that contains all the other 
nodes. In the example the degree of the minimal polynomial of $\alpha_1$, denoted $f(x)$, is $s< n-k,$ so the evaluations of $x^tf$ are dual codewords (see \eqref{eq:example-parity} in Prop.~\ref{prop:example-dual}). This implies that the downloaded symbols
form a codeword in an MDS code over $F$ which supports error correction. Importantly, this codeword does not involve the erased coordinate.

Switching to the RS codes of \cite{Tamo18} considered here, the element $\alpha_1$ is of degree $p_1$ over the repair field $F(\alpha_2,\dots,
\alpha_n)$, and generally $p_1>n-k-1,$ so the minimal polynomial of $\alpha_1$ is not a dual codeword. This requires us to modify the above idea. 
In general terms, we will find suitable elements of the set $E_1$ such that Eq.~\eqref{eq:tin} yields linear relations between the 
entries of the form $\trace(e_i v_j c_j).$ The coefficients of these relations form the rows of the matrix $M_j$.

%%%%%%%%%%%%%%%%%%%%%%%%
%
% Section CONSTRUCTING THE MATRICES
%
%%%%%%%%%%%%%%%%%%%%%%%%

\subsection{The matrices $M_j$}\label{sec:M}
In this section we will construct the matrices $M_j$ and the vector $z_j,$ and also prove the full rank condition.
Rather than writing the expressions at this point in the text, 
We proceed in stages, by deriving $p_1$ linear relations involving components of the vectors on both sides of \eqref{eq:plan}.
 (the notation is rather complicated and would not be intuitive; if desired, the reader may
nevertheless consult Sec.~\ref{sec:invertible}, particularly, Eq.\eqref{eq:z}).

\subsubsection{The first $p_1-s-1$ relations}
\begin{proposition} For all $0\le u\le s-1$ and $0\le q\le \frac{p_1-1}{s}-2$, the vector
\begin{equation} \label{eq:GRS1}
\big(\alpha_j^s \trace(\beta^u \alpha_1^{u+qs} v_j c_j) - \trace(\beta^u \alpha_1^{u+(q+1)s} v_j c_j), j=2,\dots,n \big)
\end{equation}
is a codeword in a GRS code of length $n-1$ and dimension $s+k-1.$
\end{proposition}
\begin{proof}
Let us write \eqref{eq:tin} for $e_i$ of the form $e_i=\beta^u \alpha_1^{u+qs}:$
$$
\trace(\beta^u \alpha_1^{u+qs+t} v_1 c_1) 
= -\sum_{j=2}^n \alpha_j^t \trace(\beta^u \alpha_1^{u+qs} v_j c_j) , \quad \quad
t=0,1,\dots,n-k-1
$$
(see also \eqref{eq:defE1}).
Writing this as
  $$
  \trace(\beta^u \alpha_1^{u+(q+1)s+t-s} v_1 c_1) 
= -\sum_{j=2}^n \alpha_j^{s+t-s} \trace(\beta^u \alpha_1^{u+qs} v_j c_j), \quad t=0,1,\dots,n-k-1
  $$
and performing the change of variable $(t-s)\mapsto t,$ we obtain the relation
    \begin{align}
\trace(\beta^u \alpha_1^{u+(q+1)s+t} v_1 c_1) 
= -\sum_{j=2}^n \alpha_j^{s+t} &\trace(\beta^u \alpha_1^{u+qs} v_j c_j) , \label{eq:-s}\\
&t=-s,-s+1,\dots,-s+n-k-1. \nonumber
   \end{align}
On the other hand, substitutinng $e_i=\beta^u \alpha_1^{u+(q+1)s}$ into \eqref{eq:tin}, we obtain
     \begin{equation}
\trace(\beta^u \alpha_1^{u+(q+1)s+t} v_1 c_1) 
= -\sum_{j=2}^n \alpha_j^t \trace(\beta^u \alpha_1^{u+(q+1)s} v_j c_j) , \quad \quad
t=0,1,\dots,n-k-1. \label{eq:0}
   \end{equation}
Note that the left-hand sides of \eqref{eq:-s} and \eqref{eq:0} conicide for $t=0,1,\dots,n-k-s-1,$ and thus so do the right-hand sides. We obtain
$$
\sum_{j=2}^n \alpha_j^{s+t} \trace(\beta^u \alpha_1^{u+qs} v_j c_j) =
\sum_{j=2}^n \alpha_j^t \trace(\beta^u \alpha_1^{u+(q+1)s} v_j c_j) $$
or
$$
\sum_{j=2}^n \alpha_j^t \big(\alpha_j^s \trace(\beta^u \alpha_1^{u+qs} v_j c_j) - \trace(\beta^u \alpha_1^{u+(q+1)s} v_j c_j) \big) = 0, 
$$
for $t=0,1,\dots,n-k-s-1.$
On account of Proposition~\ref{prop:RS} this implies the claim about the GRS code; moreover, since there are $n-k-s$
independent parity-check equations, the dimension of this code is $(n-1)-(n-k-s)=s+k-1$.
\end{proof}
We note that the components of the vector \eqref{eq:GRS1} are formed as linear combinations of the elements $\trace(e_i v_jc_j),$ and so this  gives us
$p_1-s-1$ vectors $z_j.$

\vspace*{.1in}
\subsubsection{{One more relation}}
\begin{proposition}
The vector
\begin{equation} \label{eq:GRS2}
\Big(\sum_{u=0}^{s-1} \alpha_j^{s-u} \trace(\beta^u \alpha_1^{u+p_1-s-1} v_j c_j) -\trace\Big(\sum_{u=0}^{s-1}\beta^u \alpha_1^{p_1-1} v_j c_j\Big),
j=2,\dots,n \Big)
\end{equation}
is a codeword in a GRS code of length $n-1$ and dimension $s+k-1$.
\end{proposition}
\begin{proof}
Going back to \eqref{eq:tin}, take $e_i=\beta^u \alpha_1^{u+p_1-s-1}$ for $u=0,1,\dots,s-1.$ We obtain the relation
   $$
\trace(\beta^u \alpha_1^{u+p_1-s-1+t} v_1 c_1) 
= -\sum_{j=2}^n \alpha_j^t \trace(\beta^u \alpha_1^{u+p_1-s-1} v_j c_j) , \quad \quad
t=0,1,\dots,n-k-1.
   $$
Changing the variable $(t+u-s) \mapsto t$ in the above equation, we obtain that for every $u=0,1,\dots,s-1$,
\begin{align} \label{eq:src}
\trace(\beta^u \alpha_1^{p_1-1+t} v_1 c_1) 
= -\sum_{j=2}^n \alpha_j^{t-u+s} &\trace(\beta^u \alpha_1^{u+p_1-s-1} v_j c_j),\\
&t=u-s,u-s+1,\dots,u-s+n-k-1 .\nonumber
\end{align}
Since
  \begin{equation}\label{eq:overlap}
\bigcap_{u=0}^{s-1} \{u-s,u-s+1,\dots,u-s+n-k-1\}
=\{-1, 0, 1, \dots, n-k-s-1\} ,
  \end{equation}
we have
\begin{align*}
\trace(\beta^u \alpha_1^{p_1-1+t} v_1 c_1) 
= -\sum_{j=2}^n \alpha_j^{t-u+s} &\trace(\beta^u \alpha_1^{u+p_1-s-1} v_j c_j), \\  &-1\le t\le n-k-s-1, \quad 0\le u\le s-1.
\end{align*}
Taking the cue from \eqref{eq:overlap}, let us 
sum these equations on $u= 0,1,\dots,s-1,$ and we obtain
    \begin{align}
\trace\Big(\sum_{u=0}^{s-1}\beta^u \alpha_1^{p_1-1+t} v_1 c_1\Big) 
= -\sum_{j=2}^n \sum_{u=0}^{s-1} \alpha_j^{t-u+s} &\trace(\beta^u \alpha_1^{u+p_1-s-1} v_j c_j), \label{eq:s}\\
   &-1\le t\le n-k-s-1. \nonumber
    \end{align}
Turning to \eqref{eq:defE1} again, let us substitute the element $\sum_{u=0}^{s-1}\beta^u\alpha_1^{p_1-1}$ into \eqref{eq:tin}:
\begin{equation} \label{eq:slp}
\trace\Big(\sum_{u=0}^{s-1}\beta^u \alpha_1^{p_1-1+t} v_1 c_1\Big) 
= -\sum_{j=2}^n \alpha_j^t \trace\Big(\sum_{u=0}^{s-1}\beta^u \alpha_1^{p_1-1} v_j c_j\Big), \quad\quad  0\le t\le n-k-1
\end{equation}
From \eqref{eq:s} and \eqref{eq:slp} we deduce the equality
    \begin{align*}
\sum_{j=2}^n \sum_{u=0}^{s-1} \alpha_j^{t-u+s} \trace\Big(\beta^u \alpha_1^{u+p_1-s-1} v_j c_j\Big)
= \sum_{j=2}^n \alpha_j^t \trace\Big(\sum_{u=0}^{s-1}\beta^u \alpha_1^{p_1-1} v_j c_j\Big),% \quad\quad  0\le t\le n-k-s-1 .
    \end{align*}
or
\begin{align*}
\sum_{j=2}^n \alpha_j^t
\Big(\sum_{u=0}^{s-1} \alpha_j^{s-u} \trace(\beta^u \alpha_1^{u+p_1-s-1} v_j c_j) -\trace\Big(\sum_{u=0}^{s-1}\beta^u \alpha_1^{p_1-1} v_j c_j
\Big) \Big)
=0 \quad\quad  
\end{align*}
for $0\le t\le n-k-s-1.$ By Proposition \ref{prop:RS}, the proof is complete.
\end{proof}

\vspace*{.1in}
\subsubsection{{The remaining $s$ relations}}
Following the plan outlined in Sec.~\ref{sec:plan}, we have constructed $p_1-s$ vectors $z_j,$ listed in \eqref{eq:GRS1} and \eqref{eq:GRS2}. 
In order to find the remaining $s$ linear combinations of the elements $r_{i,j}$, we develop the idea used in the example in Sec.~\ref{sect:toycorrection}.

We begin with introducing some notation. Let $f(x)$ be the minimal polynomial of $\alpha_1$ over $F_1.$ For $h=0,1,\dots,s-1$ define 
    \begin{equation}\label{eq:fh}
    f_h(x)=x^{p_1+h}(\text{mod} f(x)),
    \end{equation}
then $\deg f_h <\deg f=p_1$ and $\alpha_1^{p_1+h}=f_h(\alpha_1)$. Let $f_{h,q}\in F_1[x], q=0,\dots,(p_1-1)/s-1 $ be the (uniquely defined) polynomials such that 
   \begin{itemize}
   \item[(i)] $\deg f_{h,q}\le s-1, q=0,1,\dots, \frac{p_1-1}{s}-2$; 
   \item[(ii)] $\deg f_{h,(p_1-1)/s-1}\le s$;
   \item[(iii)] 
       \begin{equation} \label{eq:imp}
f_h(x)=\sum_{q=0}^{(p_1-1)/s-1} x^{qs} f_{h,q}(x)  .
       \end{equation}
   \end{itemize}

\begin{proposition}\label{prop:s}
For every $h=0,1,\dots,s-1$, the vector 
    \begin{equation} \label{eq:GRS4}
\begin{aligned}
%\begin{multline}
\Big( \sum_{q=0}^{(p_1-1)/s-1}  \sum_{u=0}^h
 f_{h-u,q}(\alpha_j) \trace(\alpha_1^{u+qs} \beta^u v_j c_j)  
 +  \sum_{u=h+1}^{s-1}  \alpha_j^{h+1-u+s} \trace(\beta^u \alpha_1^{u+p_1-s-1} v_j c_j) \\
 - \alpha_j^{h+1} \trace(\sum_{u=0}^{s-1}\beta^u \alpha_1^{p_1-1} v_j c_j), j=2,3,\dots,n
 \Big)
%\end{multline}
\end{aligned}
    \end{equation}
is contained in a GRS code of length $n-1$ and dimension $s+k-1$. 
\end{proposition}
The proof of this proposition is rather long and technical, and is given in Appendix~\ref{app:Prop4}.

Concluding, expressions \eqref{eq:GRS1}, \eqref{eq:GRS2}, and \eqref{eq:GRS4} yield $p_1$ linear combinations of the elements
$
(\trace(e_0 v_j c_j),\linebreak[3]\trace(e_1 v_j c_j),\dots,\trace(e_{p_1-1} v_j c_j)) 
$
for every $j\in\{2,3,\dots,n\}$. It is these linear combinations that we denote by
$z_j=(z_{j,0},z_{j,1}, \dots, z_{j,p_1-1})$ in \eqref{eq:plan}. We have shown that for every 
$i\in\{0,1,\dots,p_1-1\}$, the vector $(z_{2,i},z_{3,i},\dots,z_{n,i})$ is contained in an MDS code of length $n-1$ and dimension $s+k-1$.
The next subsection treats the remaining part of the assumptions in Proposition \eqref{prop:plan} above.

%%%%%%%%%%%%%%%%%%%%%%%%
%
% Section EXAMPLE: INVERTIBILITY OF TRANSFORM
%
%%%%%%%%%%%%%%%%%%%%%%%%

\subsection{The linear transforms $M_j$ are invertible}\label{sec:invertible}
The object of this section is to show that the mapping
   $$
   (\trace(e_0 v_j c_j),\trace(e_1 v_j c_j),\dots,\trace(e_{p_1-1} v_j c_j))\mapsto z_j=(z_{j,0},z_{j,1}, \dots, z_{j,p_1-1})
   $$
is invertible. In other words, we will show that $\rank(M_j)=p_1$ for all $j.$
Let us first simplify the notation.
Recall the set $E_1=\{e_0,e_1,\dots,e_{p_1-1}\}$ in \eqref{eq:defE1} and let us order its elements in the order of increase of the powers of $\alpha_1:$
      \begin{align*}
     e_{u+qs}&:= \beta^u\alpha_1^{u+qs} \text{~~for~} u=0,1,\ldots,s-1 \text{~and~} q=0,1,\ldots, {\textstyle\frac{p_1-1}s-1}\\
     e_{p_1-1}&:= \sum_{u=0}^{s-1}\beta^u\alpha_1^{p_1-1}.
       \end{align*}
Using the notation $r_{i,j}=\trace(e_i v_j c_j)$ introduced above, the vectors in \eqref{eq:GRS1} can be written as
      $$
(\alpha_j^s r_{u+qs,j} - r_{u+qs+s,j},j=2,\dots,n) \text{~~for~} 0\le u\le s-1 \text{~and~} 0\le q\le \frac{p_1-1}{s}-2.
       $$
or, writing $i=u+qs,$ as
\begin{equation} \label{eq:zz1}
(\alpha_j^s r_{i,j} - r_{i+s,j},j=2,\dots,n)
\text{~~for~} 0\le i\le p_1-s-2.
\end{equation}
Similarly, the vector in \eqref{eq:GRS2} can be written as
\begin{equation} \label{eq:zz2}
\Big(\sum_{u=0}^{s-1} \alpha_j^{s-u} r_{u+p_1-s-1,j} - r_{p_1-1,j} ,j=2,\dots,n \Big),
\end{equation}
and the vectors in \eqref{eq:GRS4} can be written as
\begin{equation} \label{eq:zz3}
\begin{aligned}
\Big( \sum_{q=0}^{(p_1-1)/s-1}  \sum_{u=0}^h
 f_{h-u,q}(\alpha_j) r_{u+qs,j}  
 +  \sum_{u=h+1}^{s-1}  \alpha_j^{h+1-u+s} r_{u+p_1-s-1,j} 
 - \alpha_j^{h+1} r_{p_1-1,j}
,j=2,\dots,n \Big) ,  \\
 0\le h\le s-1 .
\end{aligned}
\end{equation}
For a fixed value of $j$, the entries in \eqref{eq:zz1}--\eqref{eq:zz3} form the vector $z_j=(z_{j,0},z_{j,1}, \dots, z_{j,p_1-1})$, 
and we list its coordinates according to the chosen order: 
         \begin{equation}\label{eq:z}
\left.\begin{aligned}
z_{j,i} & := \alpha_j^s r_{i,j} - r_{i+s,j}
\text{~~for~} 0\le i\le p_1-s-2  , \\
z_{j,p_1-s-1} & := \sum_{u=0}^{s-1} \alpha_j^{s-u} r_{u+p_1-s-1,j} - r_{p_1-1,j} ,  \\
z_{j,p_1-s+h} & := \sum_{u=0}^h \sum_{q=0}^{(p_1-1)/s-1}  
 f_{h-u,q}(\alpha_j) r_{u+qs,j}  
 \\&\hspace*{.3in}+  \sum_{u=h+1}^{s-1}  \alpha_j^{h+1-u+s} r_{u+p_1-s-1,j} 
 - \alpha_j^{h+1} r_{p_1-1,j}
\text{~for~}  0\le h\le s-1 .
\end{aligned}\right\}
         \end{equation}
Our objective is to show that the linear mapping $(r_{0,j},r_{1,j},\dots,r_{p_1-1,j})\stackrel{M_j}\to(z_{j,0},z_{j,1}, \dots, z_{j,p_1-1})$ is invertible. 
This will follow once we show that its kernel is trivial, i.e., that if $(z_{j,0},z_{j,1}, \dots, z_{j,p_1-1})$ is the all-zeros vector, 
then so is $(r_{0,j},r_{1,j},\dots,r_{p_1-1,j}).$ If
    $
    z_{j,i}  = \alpha_j^s r_{i,j} - r_{i+s,j} =0
    $
for $0\le i\le p_1-s-2,$ then 
      \begin{equation} \label{eq:rec}
r_{u+qs,j}=\alpha_j^s r_{u+(q-1)s,j} = \dots = \alpha_j^{qs} r_{u,j}
\text{~~for~} 0\le u\le s-1 \text{~and~} 1\le q\le \frac{p_1-1}{s}-1.
      \end{equation}
Using \eqref{eq:rec} in the expression for $z_{j,p_1-s+h}, 0\le h\le s-1$, we obtain the following $s$ relations:
     \begin{align}
z_{j,p_1-s+h} & = \sum_{u=0}^h \sum_{q=0}^{(p_1-1)/s-1}  
 f_{h-u,q}(\alpha_j) \alpha_j^{qs} r_{u,j} \nonumber \\
       &\hspace*{.5in}+  \sum_{u=h+1}^{s-1}  \alpha_j^{h+1-u+s} \alpha_j^{p_1-s-1} r_{u,j} - \alpha_j^{h+1} r_{p_1-1,j}   \nonumber\\
 & = \sum_{u=0}^h   f_{h-u}(\alpha_j)  r_{u,j}  
 +  \sum_{u=h+1}^{s-1}  \alpha_j^{p_1+h-u} r_{u,j} 
 - \alpha_j^{h+1} r_{p_1-1,j} , \label{eq:zj1}
     \end{align}
where the second equality follows from \eqref{eq:imp}.
Using \eqref{eq:rec} in the expression for $z_{j,p_1-s-1}$, we obtain
   \begin{equation}\label{eq:zj2}
z_{j,p_1-s-1}  = \sum_{u=0}^{s-1} \alpha_j^{s-u} \alpha_j^{p_1-s-1} r_{u,j} - r_{p_1-1,j} = \sum_{u=0}^{s-1}  \alpha_j^{p_1-u-1} r_{u,j} - r_{p_1-1,j} .
   \end{equation}
Since we assumed that the $z$-vector is zero, coordinates $z_{p_1-u}, u=s+1,s, \dots, 1$ that appear in \eqref{eq:zj1}, \eqref{eq:zj2} are zero.
Writing these conditions in matrix form using the above order, we obtain
    \begin{equation}  \label{eq:zb}
\left[\begin{array}{cccccc}
\alpha_j^{p_1-1} & \alpha_j^{p_1-2} & \alpha_j^{p_1-3} & \dots & \alpha_j^{p_1-s} & -1 \\
f_0(\alpha_j) & \alpha_j^{p_1-1} & \alpha_j^{p_1-2} & \dots & \alpha_j^{p_1-s+1} & -\alpha_j \\
f_1(\alpha_j) & f_0(\alpha_j) & \alpha_j^{p_1-1} & \dots & \alpha_j^{p_1-s+2} & -\alpha_j^2 \\
f_2(\alpha_j) & f_1(\alpha_j) & f_0(\alpha_j) & \dots & \alpha_j^{p_1-s+3} & -\alpha_j^3 \\
\vdots & \vdots & \vdots & \vdots & \vdots & \vdots \\
f_{s-1}(\alpha_j) & f_{s-2}(\alpha_j) & f_{s-3}(\alpha_j) & \dots & f_0(\alpha_j) & -\alpha_j^s
\end{array}\right]
\left[\begin{array}{c}
r_{0,j} \\ r_{1,j} \\ r_{2,j} \\ \vdots \\ r_{s-1,j} \\ r_{p_1-1,j}
\end{array}\right] = 0 ,
     \end{equation}
We aim to show that the above matrix is invertible. 

Recall that $f(x)$ is the minimal polynomial of $\alpha_1$ and from \eqref{eq:imp}, $f(x)+f_0(x)=x^{p_1}.$
Since $f(x)$ is irreducible over $F_1$ and $\alpha_j\in F_1,$ 
we have $f(\alpha_j) \neq 0$ for all $j=2,\dots,n$.

Multiplying the first row of the matrix in \eqref{eq:zb} by $\alpha_j$ and then subtracting the second row from the first row, we obtain
$$
\left[\begin{array}{cccccc}
f(\alpha_j) & 0 & 0 & \dots & 0 & 0 \\
f_0(\alpha_j) & \alpha_j^{p_1-1} & \alpha_j^{p_1-2} & \dots & \alpha_j^{p_1-s+1} & -\alpha_j \\
f_1(\alpha_j) & f_0(\alpha_j) & \alpha_j^{p_1-1} & \dots & \alpha_j^{p_1-s+2} & -\alpha_j^2 \\
f_2(\alpha_j) & f_1(\alpha_j) & f_0(\alpha_j) & \dots & \alpha_j^{p_1-s+3} & -\alpha_j^3 \\
\vdots & \vdots & \vdots & \vdots & \vdots & \vdots \\
f_{s-1}(\alpha_j) & f_{s-2}(\alpha_j) & f_{s-3}(\alpha_j) & \dots & f_0(\alpha_j) & -\alpha_j^s
\end{array}\right].
$$
Since $f(\alpha_j) \neq 0$, we can use elementary row operations to erase the first column, obtaining
$$
\left[\begin{array}{cccccc}
f(\alpha_j) & 0 & 0 & \dots & 0 & 0 \\
0 & \alpha_j^{p_1-1} & \alpha_j^{p_1-2} & \dots & \alpha_j^{p_1-s+1} & -\alpha_j \\
0 & f_0(\alpha_j) & \alpha_j^{p_1-1} & \dots & \alpha_j^{p_1-s+2} & -\alpha_j^2 \\
0 & f_1(\alpha_j) & f_0(\alpha_j) & \dots & \alpha_j^{p_1-s+3} & -\alpha_j^3 \\
\vdots & \vdots & \vdots & \vdots & \vdots & \vdots \\
0 & f_{s-2}(\alpha_j) & f_{s-3}(\alpha_j) & \dots & f_0(\alpha_j) & -\alpha_j^s
\end{array}\right] .
$$
Proceeding analogously, let us multiply the second row of this matrix by $\alpha_j$ and then subtract the third row from the second one to obtain
   $$
\left[\begin{array}{cccccc}
f(\alpha_j) & 0 & 0 & \dots & 0 & 0 \\
0 & f(\alpha_j) & 0 & \dots & 0 & 0 \\
0 & f_0(\alpha_j) & \alpha_j^{p_1-1} & \dots & \alpha_j^{p_1-s+2} & -\alpha_j^2 \\
0 & f_1(\alpha_j) & f_0(\alpha_j) & \dots & \alpha_j^{p_1-s+3} & -\alpha_j^3 \\
\vdots & \vdots & \vdots & \vdots & \vdots & \vdots \\
0 & f_{s-2}(\alpha_j) & f_{s-3}(\alpha_j) & \dots & f_0(\alpha_j) & -\alpha_j^s
\end{array}\right] .
   $$
As above, we can eliminate all the nonzeros in the second column except for $f(\alpha_j),$ and so on. In the end we obtain the 
matrix $\text{diag}(f(\alpha_j),\dots,f(\alpha_j),-\alpha_j^s)$ with nonzero diagonal. 
This proves that the matrix in \eqref{eq:zb} is invertible. Therefore, $r_{0,j} = r_{1,j}  = \dots = r_{s-1,j} = r_{p_1-1,j}=0$. Combining this with \eqref{eq:rec}, we conclude that $r_{i,j}=0$ for all $0\le i\le p_1-1$.
This proves that the matrices $M_j, j=2,\dots,n$ in \eqref{eq:plan} are invertible, providing the last missing element to the 
justification of the repair scheme with optimal error correction.

%%%%%%%%%%%%%%%%%%%%%%%%
%
% Section OPTIMAL ACCESS
%
%%%%%%%%%%%%%%%%%%%%%%%%

	\section{A family of optimal-access RS codes}\label{sec:construction}
In this section, we construct a new family of RS codes that is similar to the construction in \cite{Tamo18} but affords repair with optimal access.

The input-output cost of node repair for the RS codes of \cite{Tamo18} was analyzed in \cite{LiDauWang19} for $d=n-1$.  
According to \eqref{eq:bandwidth}, in this case the minimum access cost per helper node equals $\frac{l}{n-k}$.  
The authors of \cite{LiDauWang19} showed that it is possible to adjust the repair scheme so that the access cost is $(1+\frac{n-k-1}{p_i})\frac{l}{n-k},$ i.e., at most twice the optimal value. However, more 
is true: namely, it turns out that any {\em fixed} node in the construction of \cite{Tamo18} (Def.~\ref{def:RSerrors}) can be repaired with optimal access. This observation,
which is the starting point of the new construction, is based on the fact that it is possible to construct a basis of the field ${\mathbb K}$ over the base field that reduces the access cost of the repair of the chosen node. If the option of choosing the basis for each erased
node were available, we could use the arguments in Sec.~\ref{sect:toyaccess} to perform repair with optimal access. The difficulty
arises because this would entail rewriting the storage contents, which should be avoided. To address this issue, we construct
the code over a field that contains $n$ elements $\beta_i$ instead of a single element $\beta$, and this supports efficient repair
of any single failed node. This idea is developed below.

\subsection{New construction}\label{sect:new}
Consider the following sequence of algebraic extensions of $\ff_p:$ let $K_0=\ff_p$ 
and for $i=1,\dots,n$ let
  \begin{equation}\label{eq:Fi}
  F_i=K_{i-1}(\alpha_i), K_i=F_i(\beta_i),
  \end{equation}
where $\alpha_i$ is an algebraic element of degree $p_i$ over $\ff_p$ and $\beta_i$ is an element of degree $s=d-k+1$ over $F_i$. In the end we obtain the field
  \begin{equation}\label{eq:K}
  {\mathbb K}:=K_n=\ff_p(\alpha_1,\dots,\alpha_n,\beta_1,\dots\beta_n). 
  \end{equation}
 We still assume that $p_1,\dots,p_n$ are distinct primes satisfying the condition $p_i\equiv 1\;\text{mod}\,s$ for all  $i=1,\dots,n.$
Consider the code $\cC:=RS_{\mathbb K}(n,k,\Omega),$ where as before, the set of evaluation points is given by
$\Omega=\{\alpha_1,\dots,\alpha_n\}.$ We will show that the code $\cC$ affords optimal-access repair.

The repair scheme follows the general approach of \cite{Guruswami16} and its implementation in \cite{Tamo18}. Let $c=(c_1,\dots,c_n)\in \cC$ be a codeword. Suppose that the node $i$ has failed (coordinate $c_i$ is erased), and we would like to repair it from a set of helper nodes $\cR\subseteq\{1,\ldots,n\}\setminus\{i\}$ with $|\cR|=d$. Let
	\begin{align*}
		h(x)=\prod_{j\in\{1,\ldots,n\}\setminus(\cR\cup\{i\})}(x-\alpha_j).
	\end{align*}
	Clearly, we have $\deg(x^t h(x))<n-k$ for $t=0,\ldots,s-1$. Therefore, for some nonzero vector $v=(v_1,\ldots,v_n)$, we have $(v_1\alpha_1^t h(\alpha_1),\ldots,v_n\alpha_n^t h(\alpha_n))\in\cC^{\bot}$ for $t=0,\ldots,s-1,$ where $\cC^\bot=GRS_{\mathbb K}(n,k,v,\Omega).$
	In other words, we have 
	\begin{align}
		\label{eq:dual}
		v_i\alpha_i^t h(\alpha_i) c_i = -\sum_{\begin{substack}{j=1\\j\neq i}\end{substack}}^{n}v_j\alpha_j^t h(\alpha_j) c_j,\quad t=0,\ldots,s-1.
	\end{align}
	The repair scheme in \cite{Tamo18} as well as in this paper
	relies on this set of $s$ dual codewords to recover the value of $c_i.$
	
\begin{remark}\label{remark:sp} The dual codewords $x^th(x)$ have zero values in the complement of the set $\hat{\cR}:=\cR\cup\{i\}.$ In other words, they are contained
in the shortened code $(\cC^\bot)^{\hat{\cR}}$ of the dual code. Thinking dually, we can start with the code $\cC^\bot$	
and construct a repair scheme for its coordinates based on the punctured code $\cC_{\hat{\cR}}$ (coordinate projection
of $\cC$ on $\hat{\cR}$). This approach is equivalent to the scheme used in
\cite{Tamo18} and in this paper because $((\cC^\bot)^{\hat{\cR}})^\bot\!\cong C_{\hat{\cR}}.$ 
\end{remark}

Let us establish a few simple properties of the tower of fields defined above in \eqref{eq:Fi}, \eqref{eq:K}.

\begin{lemma}\label{lemma:ext1} The extension degrees in the field tower $\ff_p=K_0\subset\dots\subset K_i\subset\dots\subset K_n={\mathbb K}$
are as follows:
   \begin{align*}
   [K_i:\ff_p]&=s^i\prod_{j=1}^i p_j,i=1,\dots,n\\
   [{\mathbb K}:\ff_p]&=l:=s^n\prod_{i=1}^np_i.
   \end{align*}
\end{lemma}
\begin{proof}
The proof is obvious from the definition: for each $i$ we adjoin two elements $\alpha_i,\beta_i$ to $K_{i-1},$ and their degrees 
over $K_{i-1}$ are coprime,
so they contribute $sp_i$ to the result. 
\end{proof}

We will use an explicit form of the basis of ${\mathbb K}$ over $\ff_p.$
	For $m=0,\ldots,l-1$, let us write 
	  \begin{equation}\label{eq:m}
	  m=(m_{n},m_{n-1},\ldots,m_1,\bar{m}_n,\bar{m}_{n-1},\ldots,\bar{m}_1)
	  \end{equation}
where $m_i=0,\ldots,p_i-1$ and $\bar{m}_i=0,\ldots,s-1$ for $i=1,\ldots,n$.
	
	\begin{lemma}
		\label{le:basis-a}
		Let 
		$$
		A=\{a_m:=\prod_{i=1}^{n}\alpha_i^{m_i}\prod_{j=1}^{n}\beta_j^{\bar{m}_j} \mid  m_i=0,\ldots,p_i-1,\bar{m}_j=0,\ldots,s-1;
		m=0,1,\dots,l-1\}.
		$$ 
Then $A$ is a basis for $\mathbb{K}$ over $\mathbb{F}_p$.
	\end{lemma}

	\begin{proof} By co-primality, for $i=1,\dots,n$ we have $\deg_{K_{i-1}}(\alpha_{i})=p_{i},$ and by construction, we have
$\deg_{F_i}(\beta_i)=s.$ Thus, the elements $a_m, m=0,\dots,l-1$ are linearly independent over $\ff_p$.
	\end{proof}
	
	\begin{lemma} \label{le:basis-b} 
	For $m=0,\dots,l-1$ let $\cJ=\{j\in[n]: (\bar m_j,m_j)=(s-1,p_j-1)\}$ and let
	  $$
	  b_m=\prod_{i=1}^{n}\alpha_i^{m_i}\cdot\prod_{j\in \cJ}\Big(\sum_{u=0}^{s-1}\beta_j^u \Big)\cdot\prod_{j\not\in \cJ} \beta_j^{\bar{m}_j}.
	  $$

		Then the set $B:=\{b_m\mid m=0,\ldots,l-1\}$ is a basis of $\mathbb{K}$ over $\mathbb{F}_p$. 
		
		Furthermore, for $i=1,\ldots,n$, let $A_i=\{a_m\in A \mid (m_i,\bar{m}_i)=(0,0)\}$ and $B_i=\{b_m\in B \mid (m_i,\bar{m}_i)=(0,0)\},$ then  
		$$\Span_{\mathbb{F}_p}A_i = \Span_{\mathbb{F}_p}B_i.$$
	\end{lemma}

	\begin{proof} Since $|B|=l,$ to prove that $B$ is a basis it suffices to show that the elements $a_m$ can be expressed as linear combinations
	of the elements in $B.$ Let $\cJ\subset [n]$ and let $A(\cJ)=\{a_m\in A: (\bar{m}_j,m_j)=(s-1,p_j-1),j\in \cJ;(\bar{m}_j,m_j)\neq (s-1,p_j-1),j\notin \cJ\}.$ We argue by induction on $|\cJ|.$ If $m$ is such that $\cJ=\emptyset,$ then $a_m\in B,$ and there is nothing to prove. Now assume that for all $\cJ\subset [n],|\cJ|\le J-1$ the elements $a_m$ are linearly generated by the elements in $B,$ 
and let $m$ be such that $|\cJ|=J.$ We have
  $$
  a_m=\prod_{i=1}^{n}\alpha_i^{m_i}\prod_{j\not\in \cJ}\beta_j^{\bar{m}_j}\prod_{j\in \cJ}\beta_j^{s-1}
  $$
and
  $$
  b_m=\prod_{i=1}^{n}\alpha_i^{m_i}	\prod_{j\not\in \cJ}\beta_j^{\bar{m}_j}\prod_{j\in\cJ}\sum_{u=0}^{s-1}\beta_j^u=
  \prod_{i=1}^{n}\alpha_i^{m_i}	\Big(\prod_{j\not\in \cJ}\beta_j^{\bar{m}_j}\Big)\Big(\sum_{t_1,\dots,t_J=0}^{s-1}\prod_{u=1}^J \beta_{j_u}^{t_u}
  \Big).
  $$	
Multiplying out the sums on right-hand side, we note that the term with all $t_i=s-1$ equals $a_m,$
while the remaining terms contain fewer than $J$ factors of the form $\alpha_{j_u}^{p_{j_u}-1}\beta_{j_u}^{s-1}$. Each of such
terms is contained in some $A(\cJ)$ with $ |\cJ|\le J-1,$ and is linearly generated by the elements $b_m$ by the induction hypothesis. This implies that
$a_m$ is also expressible as a linear combination of the elements in $B.$
		
To prove the second claim, note that $\Span_{\mathbb{F}_p}A_i \supseteq \Span_{\mathbb{F}_p}B_i$. Therefore, to show that $\Span_{\mathbb{F}_p}A_i = \Span_{\mathbb{F}_p}B_i$, it suffices to show that for any $a=0,\ldots,n-1$ and any $\cJ\subseteq\{1,\ldots,n\}\setminus\{i\}$, the set $A_i(\cJ)$ can be generated linearly by the set $B_i$. This proof amounts essentially to the same calculation as above, 
		and will be omitted.
	\end{proof}

The role of the basis $(b_m)$ is to eliminate as many terms on the right-hand side of \eqref{eq:dual} as possible.
To repair the node $c_i$ we use the dual basis $(b_m^\ast)$ of  $(b_m),$ writing
   \begin{align}
		c_i=v_i^{-1}\sum_{m=0}^{l-1}c_{i,m}b^{\ast}_m.\label{eq:rep}
	\end{align}
Below $\trace=\mathrm{tr}_{\mathbb{K}/\mathbb{F}_p}$ denotes the absolute trace.
	
Lemmas \ref{lemma:ext1} and \ref{le:subspace} immediately imply the following.	
	\begin{proposition} \label{coro:decomp}
For $i=1,\ldots,n$, there exists vector space $S_i$ over $K_{i-1}$ such that 
$\dim_{K_{i-1}}S_i=p_i$ and $S_i+S_i\alpha_i+\cdots+S_i\alpha_i^{s-1}=K_i$. 
Furthermore, a basis for $S_i$ over $K_{i-1}$ is given by 
		\begin{align*}
			E_i:=\{\beta_i^u\alpha_i^{u+qs}\mid u=0,\ldots,s-1;q=0,\ldots,{\textstyle\frac{p_i-1}s}-1  \} \bigcup \Big\{\alpha_i^{p_i-1}\sum_{u=0}^{s-1}\beta_i^u\Big\}.
		\end{align*}
\end{proposition}
	
We continue with the description of the repair scheme where we left in \eqref{eq:dual}. As a remark, below we write the
scheme over $\ff_p$ rather than over its extensions (the latter approach was chosen in \cite{Tamo18}).
Multiplying both sides of \eqref{eq:dual} by $\prod_{i'=1}^{n}e_{i'}\prod_{j'\neq i}^{n}\alpha_{j'}^{t_{j'}},$ where $e_{i'}\in E_{i'}$ and 
$t_{j'}=0,\ldots,s-1$, and evaluating the trace, we obtain
	\begin{align}
		\trace\Big(\prod_{i'=1}^{n}e_{i'}\prod_{j'\neq i}^{n}\alpha_{j'}^{t_{j'}} v_i\alpha_i^t h(\alpha_i) c_i\Big) 
		& = -\trace\Big(\prod_{i'=1}^{n}e_{i'}\prod_{j'\neq i}^{n}
		       \alpha_{j'}^{t_{j'}} \sum_{j\neq i}^{n}v_j\alpha_j^t h(\alpha_j) c_j\Big)\nonumber\\
		& = -\sum_{j\neq i}^{n}\trace\Big(\prod_{i'=1}^{n}e_{i'}\prod_{j'\neq i}^{n}
		            \alpha_{j'}^{t_{j'}} v_j\alpha_j^t h(\alpha_j) c_j\Big)\nonumber\\
		& = -\sum_{j\in\cR}\trace\Big(\prod_{i'=1}^{n}e_{i'}\prod_{j'\neq i}^{n}
		            \alpha_{j'}^{t_{j'}} v_j\alpha_j^t h(\alpha_j) c_j\Big).\label{eq:tr}
	\end{align}
On account of Proposition~\ref{coro:decomp} and the fact that $v_ih(\alpha_i)\neq0,$ the set 
   \begin{equation}\label{eq:basisK}
   \Big\{ \prod_{i'=1}^{n}e_{i'}\prod_{j'\neq i}^{n}\alpha_{j'}^{t_{j'}} v_i\alpha_i^t h(\alpha_i) \Big\},
   \end{equation}
   where $e_{i'}\in E_{i'}, i'\in[n];t=0,\ldots,s-1;t_{j'}=0,\ldots,s-1, j'\in[n]\backslash\{i\},$
is a basis of $\mathbb{K}$ over $\mathbb{F}_p$. Therefore, we can recover $c_i$ once we know the right-hand side of \eqref{eq:tr}. 

For $j\in\cR$, from \eqref{eq:rep} we have
	\begin{align}
		\trace(\prod_{i'=1}^{n}e_{i'}\prod_{j'\neq i}^{n}\alpha_{j'}^{t_{j'}} v_j\alpha_j^t h(\alpha_j) c_j)
		& = \trace\Big(\prod_{i'=1}^{n}e_{i'}\prod_{j'\neq i}^{n}\alpha_{j'}^{t_{j'}} \alpha_j^t h(\alpha_j) \sum_{m=0}^{l-1}c_{j,m}b^*_m\Big)\nonumber\\
		& = \sum_{m=0}^{l-1}\trace\Big(\prod_{i'=1}^{n}e_{i'}\prod_{j'\neq i}^{n}\alpha_{j'}^{t_{j'}} \alpha_j^t h(\alpha_j) b^*_m\Big)c_{j,m}.\label{eq:access}
	\end{align} 
	From \eqref{eq:access}, we see that in order to recover $c_i$ we need to access only those symbols $c_{j,m}$ for which
 $$\trace(\prod_{i'=1}^{n}e_{i'}\prod_{j'\neq i}^{n}\alpha_{j'}^{t_{j'}} \alpha_j^t h(\alpha_j) b^*_m)\neq 0.$$ 
	
Now, the element $\prod_{i'\neq i}^{n}e_{i'}\prod_{j'\neq i}^{n}\alpha_{j'}^{t_{j'}} \alpha_j^t h(\alpha_j)$ does not include $\alpha_i,\beta_i$,
and thus it can be written as an $\ff_p$-linear combination of the elements in the set $A_i.$ By 
Lemma~\ref{le:basis-b}, it can further be expressed as an $\ff_p$-linear combination of the elements in the set $B_i.$ 
Therefore, the elements $\prod_{i'=1}^{n}e_{i'}\prod_{j'\neq i}^{n}\alpha_{j'}^{t_{j'}} \alpha_j^t h(\alpha_j)$ for $e_{i'}\in E_{i'}$ and $t_{j'}=0,\ldots,s-1$ can be linearly generated over $\mathbb{F}_p$ by the set
    $$
    \bigcup_{e_i\in E_i}e_iB_i\subseteq B.
    $$
	
	Since $B$ and $B^*$ are dual bases, 
	$$
	\trace\Big(\prod_{i'=1}^{n}e_{i'}\prod_{j'\neq i}^{n}\alpha_{j'}^{t_{j'}} \alpha_j^t h(\alpha_j) b^{\ast}_m\Big)\neq 0
	$$
if and only if $b_m\in\bigcup_{e_i\in E_i}e_iB_i$. It follows that to calculate the left hand side of \eqref{eq:tr}, we need to access $\sum_{e_i\in 
E_i}|e_iB_i|=p_il/sp_i=l/s$ symbols on each helper node $j\in\cR$, which implies that the node $c_i$ affords optimal-access repair.

In conclusion, we note that the repair scheme of each of the nodes $i$ relies on its own element $\beta_i$. Looking back at the construction of 
\cite{Tamo18}, Sec.~\ref{sect:erc} above, it contains one such $\beta.$ Thus, these codes can be furnished with a repair scheme that has the optimal access property for any one (fixed) node in the encoding; see also the discussion at the end of Sec.~\ref{sect:toyaccess}.

\subsection{Error correction with optimal access}\label{sect:err-oa}
In this section we present a repair scheme of the RS codes defined in the beginning of Sec.~\ref{sect:new} that supports both the optimal access 
and optimal error correction properties. The scheme relies on a combination of ideas of Sections~\ref{sect:new} and \ref{sect:erc}. A full 
presentation of the proof would require us to repeat the arguments in Sec.~\ref{sec:M}; we shall instead confine ourselves to pointing to the similarity of the starting point and argue that once this is recognized, the remaining part is reproduced directly following the
proof in Sec.~\ref{sec:M}.

Let us modify the construction of RS codes of Sec.~\ref{sect:new} as follows. Let us assume that the number of helper nodes is $d$. 
We will construct our RS code over the symbol field $\mathbb K=\ff_p(\alpha_1,\dots,\alpha_n,\beta_1,\dots,\beta_n)$ \eqref{eq:K}, where as before, 
$\deg_{K_{i-1}}(\alpha_i)=p_i$ but $\deg_{F_i}(\beta_i)=s:=d-2e-k+1.$ Define the code $\cC:=RS_{\mathbb K}(n,k,\Omega),$ 
where $\Omega=\{\alpha_1,\dots,\alpha_n\}.$ 

Without loss of generality suppose that the failed node is the first one and let $\cR\subseteq \{2,3,\ldots,n\}$ with $|\cR|=d,2e+k\leq d\leq n-1$
be the subset of helper nodes. Consider a basis of $\mathbb{K}$ over $\ff_p$ given by $\bigcup_{t=0}^{s-1}\alpha_1^t\Lambda,$ where
    $$
   \Lambda=\Big\{ \prod_{i=1}^{n}e_{i}\prod_{j=2}^{n}\alpha_{j}^{t_{j}}  \mid  e_{i}\in E_{i}, i\in[n]; t_{j}=0,\ldots,s-1, j\in[n]\backslash\{1\}\Big\}.
   $$
That this is a basis is apparent from \eqref{eq:basisK}.

Next, note that $(v_1\alpha_1^t ,\ldots,v_n\alpha_n^t) \in\cC^{\bot}$ for some $v=(v_1,\ldots,v_n)\in({\mathbb K}^\ast)^n$ and for $t=0,\ldots,n-
k-1$. Therefore, for every $\lambda\in\Lambda$ we have
     \begin{align*}
\lambda v_1\alpha_1^t c_1 = -\sum_{j=2}^{n}  \lambda v_j\alpha_j^t c_j,\quad t=0,\ldots,n-k-1.
     \end{align*}
Let $G_1:=\ff_p(\alpha_2,\alpha_3,\ldots,\alpha_n)$. Evaluating the trace $\trace_{\mathbb{K}/G_1}$ on both sides of the above equation, we obtain
     \begin{align}
\trace_{\mathbb{K}/G_1} (\lambda v_1\alpha_1^t c_1) = -\sum_{j=2}^{n} \alpha_j^t\trace_{\mathbb{K}/G_1} (\lambda v_j c_j),\quad t=0,\ldots,n-
       k-1.\label{eq:tr-1}
     \end{align}
The repair scheme for the code $\cC$ is based on \eqref{eq:tr-1} in exactly the same way as the repair scheme of Proposition~\ref{prop:plan} 
is based on \eqref{eq:tin}. Namely, suppose that there are invertible linear transformations that map the vectors $(\trace_{\mathbb{K}/G_1} 
(\lambda v_j c_j),\lambda\in\Lambda),j=2,3,\ldots,n$ to codevectors in an MDS code of length $n-1$ and dimension $s+k-1$. Then it is possible to 
correct $e$ errors in the information collected from the helper nodes upon puncturing of this code to any $d$ coordinates in the same way as is 
done in Proposition~\ref{prop:plan}. Thus, the main step is to prove existence of such transformations. 
Here we observe that the terms involved in \eqref{eq:tr-1} are formed of $e_1$ times the remaining factors in $\lambda.$ The element $e_1$
plays the same role as $e_i$ in \eqref{eq:tin}, and the multiplier in front of it in $\lambda$ does not affect the proof. For this reason, the
required proof closely follows the proof in Sec.~\ref{sec:M}, and we do not repeat it here.

Thus, the vectors $(\trace_{\mathbb{K}/G_1} (\lambda v_j c_j),\lambda\in\Lambda),j\in\cR$ suffice to 
recover the value of the failed node. We argue that these values can be calculated by accessing the smallest possible number of symbols on the
helper nodes, and thus support the claim of optimal access. Let $B=(b_m)$ be the basis of $\mathbb K$ over $\ff_p$ defined in 
Lemma~\ref{le:basis-b}, let $B^\ast=(b^\ast_m)$ be its dual basis, and let $B_1=\{b_m\in b|(m_1,\bar m_1)=(0,0)\}.$ 
From \eqref{eq:rep}, for every $\lambda\in\Lambda$ and all $j=2,3,\ldots,n$ we have the equality
\begin{align*}
	\trace_{\mathbb{K}/G_1} (\lambda v_j c_j) = \trace_{\mathbb{K}/G_1} \Big(\lambda \sum_{m=0}^{l-1}b^{\ast}_m\Big)c_{i,m}.
\end{align*}
Let $\Gamma$ be a basis for $G_1$ over $\ff_p$. Then from the above equation, for every $\gamma\in\Gamma$ we have
\begin{align*}
	\trace_{G_1/\ff_p}(\gamma\trace_{\mathbb{K}/G_1} (\lambda v_j c_j)) 
	= \trace_{G_1/\ff_p}\Big(\gamma\trace_{\mathbb{K}/G_1} \Big(\lambda \sum_{m=0}^{l-1}b^{\ast}_m\Big)\Big)c_{i,m}.
\end{align*}
Since $\gamma\in G_1$ and $\trace_{G_1/\ff_p}\circ\trace_{\mathbb{K}/G_1}=\trace_{\mathbb{K}/\ff_p}$, it follows that
\begin{align}
	\trace_{\mathbb{K}/\ff_p} (\gamma\lambda v_j c_j) = \trace_{\mathbb{K}/\ff_p} \Big(\gamma\lambda \sum_{m=0}^{l-1}b^{\ast}_m\Big)c_{i,m}.\label{eq:down}
\end{align}
Note that the elements $\gamma\lambda = \gamma\prod_{i=1}^{n}e_{i}\prod_{j=2}^{n}\alpha_{j}^{t_{j}}$ can be written as $\ff_p$-linear combinations of the elements in the set
$\bigcup_{e_1\in E_1}e_1B_1\subseteq B$. By the duality of $B$ and $B^*$, the number of symbols that each helper node accesses to calculate the left hand side of \eqref{eq:down} equals $|\bigcup_{e_1\in E_1}e_1B_1|=l/s$, which, as remarked in the introduction, is the smallest possible number of symbols. Further, since $\Gamma$ is a basis of $G_1$ over $\ff_p$, we can recover $\trace_{\mathbb{K}/G_1}(\lambda v_j c_j)$ from the set $\{\trace_{\mathbb{K}/\ff_p} (\gamma\lambda v_j c_j)\mid \gamma\in\Gamma\}$.

Finally, evaluating the trace $\trace_{G_1/\ff_p}$ on both sides of \eqref{eq:tr-1}, we obtain 
\begin{align}
	\trace_{\mathbb{K}/\ff_p} (\lambda v_1\alpha_1^t c_1) = -\sum_{j=2}^{n} \trace_{G_1/\ff_p}(\alpha_j^t\trace_{\mathbb{K}/G_1} (\lambda v_j c_j)),\quad t=0,\ldots,s-1.\label{eq:tr-2}
\end{align}
Since the set $\{\lambda v_1\alpha_1^t\mid \lambda\in\Lambda;t=0,\ldots,s-1 \}$ forms a basis for $\mathbb{K}$ over $\ff_p$, we conclude from \eqref{eq:tr-2} that we can perform optimal error correction for the code $\cC$ with optimal access. As a final remark, the locations
of the entries accessed on each helper node depend only on the index of the failed node, and are independent of the index of the helpers.

%%%%%%%%%%%%%%%%%%%%%%%%
%
% Section ALGORITHM
%
%%%%%%%%%%%%%%%%%%%%%%%%

	\section{Every scalar MSR code affords optimal-access repair}
%	\label{sec:constant-subspace}
This section is devoted to establishing the claim in the title. We begin with a discussion of repair schemes with a particular property of
having constant repair subspaces and use it to show that every MSR code with this property can be repaired with optimal access. In the last
part of the section we remove this assumption, establishing the general result, which is stated as follows.
	\begin{theorem} \label{thm:OA}
Let $\cC$ be an $(n,k)$ scalar MDS code over a finite field $K$ of length $n$ such that any single failed node can be optimally repaired from 
any subset of $d$ helper nodes, $k+1\le d\le n-1$ with optimal repair bandwidth. Then there exists an explicit procedure that  supports 
optimal-access repair of any single node from any subset of $d$ helpers, $k+1\le d\le n-1$.
	\end{theorem}

\subsection{Constant repair subspaces}\label{subsec:crs}
Observe that the repair scheme presented above in Sec.~\ref{sec:construction} has the property that for a given index of the failed node $i$, the procedure for recovering the node contents does not depend on the chosen subset of $d$ helper nodes. Indeed, to repair node $i$, 
the scheme accesses symbols $\{c_{j,m}\mid m:b_m\in\bigcup_{e_i\in E_i}e_iB_i \}$ on the node $j$, i.e., the symbols $c_{j,m}$ with 
$m=(m_i,\bar m_i)$ and
      $$
    (m_i,\bar{m}_i)\in\{(u+qs,u)\;|\;u=0,\ldots,s-1;q=0,\ldots,(p_i-1)/s-1\}\cup\{(p_i-1,s-1) \}.
      $$
Clearly the values of $m$ are independent of $j\in\cR$. This simplifies the implementation, and therefore represents a
desirable property of the scheme.
In this section, we generalize this observation and give conditions for it to hold.

Let $\cC$ be an $(n,k)$ linear scalar MDS code of length $n$ over finite field $K,$ and let $r=n-k$ be the number of parity nodes. Let $F$ be a subfield of $K$ such that $[K:F]=l$. For a subset $M\subset K$ we write $\dim_F(M)$ to refer to the dimension of the subspace spanned
by the elements of $M$ over $F$. The following result is a starting point of our considerations.
	
	\begin{theorem}[\cite{Guruswami16}]
		\label{thm:scheme}
		The code $\cC$ has an optimal linear repair scheme over $F$ with repair degree $d=n-1$ if and only if for every $i=1,\ldots,n$ there exist $l$ codewords $(c^{\perp}_{t,1},\ldots,c^{\perp}_{t,n})\in \cC^\perp,t=1,\ldots,l$ such that 
		\begin{align*}
		\dim_{F}(c^{\perp}_{1,i},\ldots,c^{\perp}_{l,i})&=l,\\
		\sum_{j\neq i}^{n}\dim_{F}(c^{\perp}_{1,j},\ldots,c^{\perp}_{l,j})&=\frac{(n-1)l}{r}.
		\end{align*}
	\end{theorem}

We go on to define the main object of this section.

	\begin{definition}\label{def:crc1} Let $\cC$ be a scalar MDS code that has a linear repair scheme for repair of a single node with optimal 
bandwidth, based on dual codewords $c^{\perp}_{1},\ldots,c^{\perp}_{l}.$ The scheme is said to have constant repair subspaces if for every $i=1,\ldots,n$ and every $\cR\subset [n]\backslash\{i\},|\cR|=d$, the information downloaded from a helper node 
$c_j,j\in\cR$ to repair the failed node $c_i$ does not depend on the index $j$. Namely, the subspace
$\cS_j^{(i)}:=\Span_{F}(c^{\perp}_{1,j},\ldots,c^{\perp}_{l,j}),j\in \cR$ is independent of the index $j,$ i.e.,
$  \cS_{j}^{(i)}=\cS^{(i)}$    for some linear subspace $\cS^{(i)}\subseteq K.$
    \end{definition}

The notion of constant repair subspaces was mentioned earlier in the literature on general MSR codes, for instance, see \cite{Tamo14}. 
	
The algorithms below in this section rely on a proposition which we cite from \cite{Tamo18}.
\begin{proposition}
		\label{coro:scheme} Let $\cC$ be an $(n,n-r)$ MDS code and let $[n]=J\cup J^c,$ where $J,|J|=r$ is the set of parity coordinates. 
Let $H=(h_1,\dots,h_n)$ be a parity-check matrix of $\cC,$ where $h_i$ denote its columns.
The code $\cC$ has an optimal linear repair scheme over $F$ with repair degree $d=n-1$ if and only if for each 
$j\in J^c$ there exist $r$ vectors $a_u\in K^{l/r},u=1,\ldots,r$ such that 
		\begin{align}
		\dim_{F}(Ah_j)&=l,\label{eq:subspace-failed}\\
		\dim_{F}(Ah_i)&=\frac{l}{r},\quad i\in \{1,\ldots,n\}\setminus \{j\},\label{eq:subspace-helper1}
		\end{align}
where $A:=\mathrm{Diag}(a_1,\ldots,a_r)$ is an $l\times r$ block-diagonal matrix
with blocks formed by single columns.
		Furthermore for every subspace $\cA_u=\Span_{F}(a_u), u=1,\dots, r$ (the $F$-linear span of the entries of $a_u$) we have 
		\begin{align}
		\dim_{F}(\cA_u) = \frac{l}{r}.\label{eq:subspace-dim}
		\end{align}
	\end{proposition}
	
	\begin{remark}
The matrix $A$ in Proposition~\ref{coro:scheme} depends on the matrix $H$ and the choice of $J$, but we 
suppress this dependence from the notation for simplicity.
	\end{remark}
	
Before presenting the algorithms for finding a basis for optimal-access repair we briefly digress to state some conditions for an optimal 
linear repair scheme to have constant repair subspaces. First, we rephrase their definition based 
Proposition~\ref{coro:scheme}. 
	
	\begin{definition}\label{def:crs}
		An optimal linear repair scheme for the code $\cC$ is said to have constant repair subspaces if for every $j=1,\ldots,n$ there exists a vector $h\in K^r$ such that 
		$$\Span_F(Ah_i)=\Span_F(Ah)$$ 
for every $i\in\{1,\ldots,n\}\setminus\{j\}$. Here the matrix $A$ is as in Proposition~\ref{coro:scheme}, and it depends on $H$ and the 
particular choice of the information coordinates.
	\end{definition}
	
	\begin{proposition}
		\label{prop:constant-sufficient}
Suppose that $\cA_1=\cA_2=\cdots=\cA_r$ for each $i=1,\ldots,n,$ and that for every $j\in\{1,\ldots,n\}\setminus\{i\}$ there exists $v\in \{1,\ldots,r\}$ such that $h_{v,j}\in F$, then there exists an optimal linear repair scheme for the code $\cC$ which has constant repair subspaces.
	\end{proposition}
	
	\begin{proof}
		Let $\cV$ denote any of the (coinciding) repair subspaces. By Proposition~\ref{coro:scheme}, we have $\dim_{F}(\cV )=l/r$. 
Suppose that $J$ is the subset of parity coordinates, and the matrix $H$ is represented in systematic form. In this case, 
for every $j\in J^c,$ $h_{u,j}\neq 0$ for all $u=1\ldots,r,$ and we have $\dim_{F}(\cV h_{u,j})=l/r$. Note that 
		\begin{align}
		\Span_F(Ah_j)=\sum_{u=1}^{r}\cA_uh_{u,j}=\sum_{u=1}^{r}\cV h_{u,j},\quad j\in\{1,\ldots,n\}\setminus\{i\},\label{eq:subspace-sum}
		\end{align}
where the sum on the right is a sum of linear spaces.
		By Proposition~\ref{coro:scheme}, we also have $l/r=\dim_{F}(Ah_j)=\dim_{F}(\sum_{u=1}^{r}\cV h_{u,j})$. Therefore,
		\begin{align}
		\cV h_{1,j} = \cV h_{2,j} = \cdots = \cV h_{r,j},\quad j\in J^c\setminus\{i\}.\label{eq:collapse}
		\end{align}
		Since for each $j\ne i$ there exists $v\in \{1,\ldots,r\}$ such that $h_{v,j}\in F,$ it follows that $\cV h_{v,j}=\cV $. On account of \eqref{eq:subspace-sum} and \eqref{eq:collapse}, we have $\Span_F(Ah_j)=\cV =\Span_F(A\cdot\mathbf{1})$ for every $j\in\{1,\ldots,n\}\setminus\{i\}$, where $\mathbf{1}$ is the all-ones column vector of length $r$. By Definition \ref{def:crs} this completes the proof.
	\end{proof}
The assumptions of this proposition are satisfied, for instance, for the RS subfamily of \cite{Tamo18}, which therefore have
constant repair subspaces (this observation was previously not stated in published literature).
	
	\begin{proposition}
		\label{prop:constant-necessary}
		If there exists an optimal linear repair scheme for the code $\cC$ which has constant repair subspaces, then $\cA_1=\cA_2=\cdots=\cA_r$ for every $j=1,\ldots,n$.
	\end{proposition}

	\begin{proof}
		Indeed, since $H_J$ is the identity, for $j\in J$ we have $\Span_F(Ah_j)=\cA_t$ for some $t\in\{1,\ldots,r\}$. It follows that $\cA_1=\cA_2=\cdots=\cA_r$.
	\end{proof}
	
\subsection{Optimal access for the case of constant repair subspaces}\label{sec:crs}
The codes constructed in Sec.~\ref{sec:construction} above form essentially the only known example of RS codes that afford repair
with optimal access. For instance, the optimal-repair RS codes in \cite{Tamo18} are not known to support optimal access, and the repair scheme in
\cite{Tamo18} is far from having this property. Prior works on the problem of access cost for RS repair \cite{DauDuursmaChu18,DauViterbo18,LiDauWang19} also do not give examples of repair schemes with optimal access.
In this section we show that any family of scalar MDS codes with optimal repair can be furnished with a repair scheme with optimal access, and this includes the code family in \cite{Tamo18}. Unfortunately, our results are not explicit; rather, we present an algorithm that produces 
a basis for representing nodes of the codeword that supports optimal-access repair.
	
	As in Sec.~\ref{subsec:crs}, let $F$ be a subfield of $K$ such that $[K:F]=l$.
Let $\cC$ be an $(n,k=n-r)$ linear scalar MDS code of length $n$ over $K$ equipped with a repair scheme over $F$ that attains the bound
\eqref{eq:bandwidth} for repair of a single node. Let us represent $\cC$ in systematic form, choosing a subset $J\subseteq\{1,\ldots,n\},|J|=r$
for the parity symbols and $J^c$ for the data symbols. Let $H$ be an $r\times n$ parity-check matrix for $\cC$ such that $H_J$ is the $r\times r$ identity matrix, 

In this section we assume that there exists an optimal repair scheme over $F$ for $\cC$ that
has constant repair subspaces, and that the repair degree is $d=n-1.$ We will lift both assumptions and show that our result holds 
in general in the next section. For a given $j=1,\dots, n$ consider the subspaces $\cA_i, i=1,\dots,r$ defined in Proposition~\ref{coro:scheme}. Under the assumption of constant repair subspaces, they coincide, and we use the notation $\cV_j$ to refer to any of them.

Consider the following procedure (Algorithm~\ref{alg:basis}) that interatively collects vectors to form a basis of $K/F$ that supports optimal-access 
repair. 
	\begin{algorithm}
		\KwIn{Subspaces $\cV_1,\ldots,\cV_n$.}
		\KwOut{A basis $B$ for $K$ over $F$.}
		%initialization\;
		\For{$j\leftarrow 1$ \KwTo $n$}{
			$B_j\leftarrow \emptyset$\;
			$\cB_j\leftarrow \{0\}$\;
		}
		\For{$i\leftarrow 0$ \KwTo $n-1$}{
			\ForEach{$I\subseteq \{1,\ldots,n\}$ such that $|I|=i$}
			{
					$\bar{I}\leftarrow \{1,\ldots,n\}\setminus I$\;
					$\cU_I\leftarrow\bigcap_{j\in\bar{I}}\cV_j$\;
					\For{$j\leftarrow 1$ \KwTo $n$}
					{
						\If{$j\in \bar{I}$}
						{
							$\cB_j\leftarrow \cB_j+\cU_I$\;
							Extend the set $B_j$ to a basis of $\cB_j$ over $F$\;
						}
					}
			}
		}
		$\bar{B}\leftarrow \bigcup_{j=1}^{n}B_j$\;
		Extend the set $\bar{B}$ to a basis $B$ of $K$ over $F$\;
		\caption{Construction of an optimal basis}
		\label{alg:basis}
	\end{algorithm}

	\begin{proposition}
		\label{prop:bj}
		Upon completion of Algorithm~\ref{alg:basis} we have $\cB_j=\cV_j$ for $j=1,\ldots,n$, and thus $B_j$ is a basis for $\cV_j$ over $F$.
	\end{proposition}
	
	\begin{proof}
		From Algorithm~\ref{alg:basis}, we have
		\begin{align}
			\cB_j=\sum_{i=0}^{n-1}\sum_{\substack{|I|=i,\\I\subseteq\{1,\ldots,n\}}}
			\mathbbm{1}_{\{j\in \bar{I}\}}
			\bigcap_{t\in \bar{I}}\cV_t,
		\end{align}
so clearly $\cB_j\subseteq \cV_j$. Suppose that $v\in\cV_j\backslash \cB_j$, then there exists a subset $\bar{I}\subset\{1,\dots,n\}$ with $1\leq |\bar{I}|\leq n$ such that $j\in\bar{I}$ and that
		\begin{align*}
			v\notin \bigcap_{t\in \bar{I}}\cV_t.
		\end{align*}
		However, $\cB_j\supseteq \bigcap_{t\in \bar{I}}\cV_t$ for every $\bar{I}$ with $1\leq |\bar{I}|\leq n$ such that $j\in\bar{I}$, which is a contradiction. Hence, $\cB_j=\cV_j$.
	\end{proof}

	\begin{proposition}
		\label{prop:b}
		Algorithm~\ref{alg:basis} returns a basis $B$ for $K$ over $F$.
	\end{proposition}

	\begin{proof}
		From Algorithm~\ref{alg:basis}, for every $\bar{I}\subseteq \{1,\ldots,n\}$ with $1\leq |\bar{I}|\leq n$
and for every $j\in\bar I$, the set $B_j$ contains a basis of the subspace $\cU_I=\bigcap_{t\in\bar{I}}\cV_t.$ It follows that for every $\bar{I}\subseteq \{1,\ldots,n\}$ with $1\leq |\bar{I}|\leq n$, the set $\bigcap_{t\in\bar{I}}B_t$ is a basis for $\bigcap_{t\in\bar{I}}\cV_t$.
		
		Now by Proposition~\ref{prop:bj}, $B_1,B_2$ are bases for $\cV_1,\cV_2$ over $F$, respectively. From the above, we have $B_1\cap B_2$ is a basis of $\cV_1\cap\cV_2$ over $F$. It follows that $\dim_{F}(\cV_1\cap\cV_2)=|B_1\cap B_2|$. Then
		\begin{align*}
			\dim_{F}(\cV_1+\cV_2) &= \dim_{F}(\cV_1)+\dim_{F}(\cV_2)-\dim_{F}(\cV_1\cap\cV_2)\\
			&=|B_1|+|B_2|-|B_1\cap B_2|\\
			&=|B_1\cup B_2|.
		\end{align*}
By definition, $\Span_F(B_1\cup B_2)=\cV_1+\cV_2$, and so the set $B_1\cup B_2$ is a basis of $\cV_1+\cV_2$ over $F$. By a straightforward induction argument, we conclude that $\bigcup_{j=1}^{n}B_j$ is a basis for $\sum_{j=1}^{n}\cV_j$ over $F$.
		
		Since $\sum_{j=1}^{n}\cV_j\subseteq K$, we have $|\bigcup_{j=1}^{n}B_j|\leq [K:F]=l$. It follows that we can extend the set $\bar{B}=\bigcup_{j=1}^{n}B_j$ to a basis $B$ of $K$ over $F$.
	\end{proof}
	
Now we are ready to present a repair scheme for the code $\cC$ with the optimal access property. Let $B=(b_m)$ be the basis of $K$ over $F$
constructed above and let $B^{\ast}=(b_m^{\ast})$ be its dual basis. Given a codeword $c=(c_1,\ldots,c_n)\in\cC,$ we expand
its coordinates in the basis $B^{\ast},$ writing
	\begin{align}
		c_i=\sum_{m=0}^{l-1}c_{i,m}b^{\ast}_m.
	\end{align}
Suppose that $c_i$ is the erased coordinate of $c$ (the ``failed node''). The starting point, as above, is Eq.~\eqref{eq:dual}, and 
our first step is to choose $l$ dual codewords $c^{\bot}_t, t=1,\dots,l$ that support the repair. Construct the $l\times n$ 
matrix $C^{\bot}=AH$ and take the rows of $C$ to be the needed codewords $c^{\bot}_t.$
Since $c_t^\bot \cdot c=0$ for all $t$, we have $c_{t,i}^\bot c_i=-\sum_{\begin{substack}{j=1\\j\ne i}\end{substack}}^n c_{t,j}^\bot c_j$
for all $t=1,\ldots,l.$ Computing the trace $\trace_{K/F}$, we obtain
	\begin{align}
		\trace_{K/F}(c^{\perp}_{t,i}c_i) & = -\sum_{j\neq i}^{n}\trace_{K/F}(c^{\perp}_{t,j}c_j)\nonumber\\
		& = -\sum_{j\neq i}^{n}\trace_{K/F}(c^{\perp}_{t,j}\sum_{m=0}^{l-1}c_{j,m}b^{\ast}_m)\nonumber\\
		& = -\sum_{j\neq i}^{n}\sum_{m=0}^{l-1}\trace_{K/F}(c^{\perp}_{t,j} b^{\ast}_m)c_{j,m}.\label{eq:tr-constant-1}
	\end{align}
Note that for each $j\in\{1,\ldots,n\}\setminus\{i\}$, we have
	\begin{align}
		\Span_{F}(c^{\perp}_{1,j},\ldots,c^{\perp}_{l,j}) = \Span_{F}(Ah_j) =\cV_i,
	\end{align}
where the last equality follows by the assumption of constant repair subspaces. 
By Proposition~\ref{prop:bj}, the set $B_i\subseteq B$ is a basis for $\cV_i$ over $F$. Therefore, $c^{\perp}_{t,j}$ can be linearly generated by the set $B_i$ for every $t=1,\ldots,l$. More precisely, let $B_i=\{b_{i,u}\,|\, u=1,\ldots,l/r\},$ then we have
	\begin{align}
		c^{\perp}_{t,j} = \sum_{u=1}^{l/r}\gamma_{j,u} b_{i,u}\label{eq:bj}
	\end{align}
for some $\gamma_{j,u}, u=1,\dots,l/r.$	Substituting into \eqref{eq:tr-constant-1}, we obtain the equality
	\begin{align}
		\trace_{K/F}(c^{\perp}_{t,i}c_i) = -\sum_{j\neq i}^{n}\sum_{m=0}^{l-1}\sum_{u=1}^{l/r}\trace_{K/F}(b_{i,u} b^{\ast}_m)\gamma_{j,u} c_{j,m}.\label{eq:tr-constant-2}
	\end{align}
	It follows that to determine the left-hand side of \eqref{eq:tr-constant-2}, on each node $c_j,j\neq i$ the repair procedure
needs to access the set of symbols $\{ c_{j,m}\,|\, \trace_{K/F}(b_{i,u} b^{\ast}_m) = 1 \}.$ Since $B_i\subseteq B$ and $B^{\ast}$ is the dual basis of $B$ for $K$ over $F,$ the cardinality of this subset equals $|B_i|=l/r,$ verifying that the repair
can be accomplished with the minimum possible access cost.

\subsection{Optimal-access repair for general scalar MSR codes}\label{sec:grs}
In this section we extend the above arguments for optimal repair schemes that do not necessarily have constant repair subspaces. 
This is done by a simple extension of Algorithm~\ref{alg:basis}. We use the same notation as in Sec.~\ref{sec:crs}.
	
\vspace*{.1in}
\noindent\subsubsection{Repair degree $d=n-1$}
\begin{equation*}\end{equation*}

\vspace*{-.2in}
\noindent
Assume that the index of the failed node is $i\in\{1,\dots,n\}.$ By Proposition~\ref{coro:scheme}, for each $j\in\{1,\ldots,n\}\setminus\{i\}$, we have
	\begin{align*}
		\dim_{F}(\cA_u)=\dim_{F}(Ah_j)=\frac{l}{r},\quad u=1,\ldots,r.
	\end{align*}
It follows that for $j\in J^c\setminus\{i\}$ we have
	\begin{align*}
		\cA_1h_{j,1}=\cA_2h_{j,2}=\cdots=\cA_rh_{j,r}.
	\end{align*}
Let $J=(i_1,\dots,i_r)$ be the set of parity nodes written in increasing order of their indices, and for $i_t\in J$ let $\sigma(i_t)=t.$ Define
	\begin{align}\label{eq:sigma}
		\cV_i^{(j)}=
		\begin{cases}
		\cA_1h_{j,1} & j\in J^c\setminus\{i\},\\
		\cA_{\sigma(j)} & j\in J.
		\end{cases}
	\end{align}
	
	\begin{algorithm}
%		\SetAlgoLined
		\KwIn{Subspaces $\cV_i^{(j)},i\in\{1,\ldots,n\},j\in\{1,\ldots,n\}\setminus\{i\}$.}
		\KwOut{A basis $B$ for $K$ over $F$.}
		%initialization\;
		\For{$i\leftarrow 1$ \KwTo $n$}{
			\ForEach{$j\in\{1,\ldots,n\}\setminus\{i\}$}{
				$B_i^{(j)}\leftarrow \emptyset$\;
				$\cB_i^{(j)}\leftarrow \{0\}$\;
			}
		}
		$\Omega\leftarrow \{1,\ldots,n\}^2\setminus\{(i,i)\mid i=1,\ldots,n\}$\;
		\For{$u\leftarrow 0$ \KwTo $n^2-n-1$}{
			\ForEach{$I\subseteq \Omega$ such that $|I|=u$}{
				$\bar{I}\leftarrow \Omega\setminus I$\;
				$\cU_I\leftarrow\bigcap_{(i,j)\in\bar{I}}\cV_i^{(j)}$\;
				\For{$i\leftarrow 1$ \KwTo $n$}{
					\ForEach{$j\in\{1,\ldots,n\}\setminus\{i\}$}{
						\If{$(i,j)\in \bar{I}$}{
							$\cB_i^{(j)}\leftarrow \cB_i^{(j)}+\cU_I$\;
							Extend the set $B_i^{(j)}$ to be a basis of $\cB_i^{(j)}$ over $F$\;
						}
					}
				}
			}
		}
		$\bar{B}\leftarrow \bigcup_{i=1}^{n}\bigcup_{j\neq i}^{n}B_i^{(j)}$\;
		Extend the set $\bar{B}$ to be a basis $B$ for $K$ over $F$\;
		\caption{Construction of an optimal basis; repair degree $d=n-1$}
		\label{alg:basis-2}
	\end{algorithm}
	
	\begin{proposition}
		\label{prop:bj-2}
When Algorithm~\ref{alg:basis-2} terminates, we have $\cB_i^{(j)}=\cV_i^{(j)}$ for $i\in\{1,\ldots,n\}$ and 
$j\in\{1,\ldots,n\}\setminus\{i\}$, and thus $B_i^{(j)}$ is a basis for $\cV_i^{(j)}$ over $F$.
	\end{proposition}
	
	\begin{proposition}
		\label{prop:b-2}
		Algorithm~\ref{alg:basis-2} returns a basis $B$ for $K$ over $F$.
	\end{proposition}

The proofs of Propositions~\ref{prop:bj-2} and \ref{prop:b-2} follow closely the proofs of Proposition~\ref{prop:bj} and \ref{prop:b}
and will be omitted.

Now it is not difficult to see that we can repair the failed node $c_i$ with optimal access cost relying on the basis $B$.
 Indeed, for each $j\in\{1,\ldots,n\}\setminus\{i\}$, we have
	\begin{align}
		\Span_{F}(c^{\perp}_{1,j},\ldots,c^{\perp}_{l,j}) = \Span_{F}(Ah_j) =\cV_i^{(j)}.
	\end{align}
	By Algorithm~\ref{alg:basis-2} and Proposition~\ref{prop:bj-2}, the set $B_i^{(j)}\subseteq B$ is a basis for $\cV_i^{(j)}$ over $F$. Therefore, $c^{\perp}_{t,j}$ can be linearly generated by the set $B_i^{(j)}$ for every $t=1,\ldots,l$. Let $B_i^{(j)}=\{b_{i,u}^{(j)}\mid u=1,\ldots,l/r \}$. Then, similarly to \eqref{eq:bj} and \eqref{eq:tr-constant-2}, we have
	\begin{align}
		c^{\perp}_{t,j} &= \sum_{u=1}^{l/r}\gamma_{j,u} b_{i,u}^{(j)},\\
		\trace_{K/F}(c^{\perp}_{t,i}c_i) &= -\sum_{j\neq i}^{n}\sum_{m=0}^{l-1}\sum_{u=1}^{l/r}\trace_{K/F}(b_{i,u}^{(j)} b^{\ast}_m)\gamma_{j,u} c_{j,m}.
	\end{align}
	Therefore, each node $c_j,j\neq i$ needs to access the set of symbols $\{ c_{j,m}\mid \trace_{K/F}(b_{i,u}^{(j)} b^{\ast}_m) = 1 \}$, whose cardinality is given by $|B_i^{(j)}|=l/r$. It follows that the repair scheme has the optimal access property.
	
\vspace*{.1in}\subsubsection{Arbitrary repair degree}
\begin{equation*}\end{equation*}

\vspace*{-.2in}
	So far we assumed that the repair relies on all the surviving nodes except for the single failed node, i.e., $|\cR|=n-1.$
In this section we derive the most general version of the result of this section, that any scalar MDS code can be repaired
with optimal access from any subset of helper nodes $\cR$ of size $d, k+1\le d\le n-1.$
Let $s:=d-k+1$.
	
Let $G=[g_1 | g_2 | \ldots | g_n]$ be a $k\times n$ generator matrix of $\cC,$ where $g_i$ is a $k$-column over $K$. Let $i\in\{1,\dots, n\}$ and 
let $\cR\not\ni \{i\}$ be a subset of $d$ helper nodes. Let $\hat{\cR}=\cR\cup\{i\}$ and $G_{\hat{\cR}}$ be the $k\times (d+1)$ submatrix formed 
by the columns $g_j,j\in\hat{\cR}$. Clearly, $G_{\hat{\cR}}$ defines a $(d+1,k)$ punctured code $\cC_{\hat{\cR}}$ of the code $\cC$. 
Since $\cC$ is MDS, the code $\cC_{\hat{\cR}}$ is itself MDS. 
Let $H^{\hat{\cR}}=(h_i^{(\hat{\cR})}, i=1,\dots,d+1)$ be an the $s\times (d+1)$ parity-check matrix of the code $\cC_{\hat{\cR}}$. Recalling Remark~\ref{remark:sp},
the code generated by $H^{\hat\cR}$ is a shortened code $(\cC^{\bot})^{\hat{\cR}},$
i.e., a subcode of $\cC^\bot$ formed of the codewords with zeros in the coordinates in $\hat{\cR}^c$.
	
Suppose that the code $\cC$ can optimally repair any single failed node $i$ from the coordinates in $\cR=\hat \cR\backslash\{i\}.$ This means that the MDS 
code $\cC_{\hat{\cR}}$ can optimally repair any single failed node $i$ from the helper 
nodes $\hat{\cR}\setminus\{i\}$.  Let $J\subseteq \hat{\cR}, |J|=s$ and $i\notin J$ and assume without loss of generality that 
the submatrix $H^{\hat{\cR}}_J$ is an $s\times s$ identity matrix. 
Now Proposition~\ref{coro:scheme} applied for the code $\cC_{\hat{\cR}}$ guarantees that there exist vectors $a_u\in K^{l/s},u=1,\ldots,s$ such that the block-diagonal
matrix $A=\mathrm{Diag}(a_1,\ldots,a_s)$ satisfies
	\begin{align}
		\dim_{F}(Ah_i^{(\hat{\cR})})&=l,\label{eq:subspace-failed-s}\\
		\dim_{F}(Ah_j^{(\hat{\cR})})&=\frac{l}{s},\quad j\in \hat{\cR}\setminus\{i\},\label{eq:subspace-helper1-s}\\
		\dim_{F}(\cA_u)&=\frac{l}{s},\quad u=1,\ldots,s,\label{eq:subspace-dim-s}
	\end{align}
	where $\cA_u:=\Span_{F}(a_u)$.
	
	It follows from \eqref{eq:subspace-helper1-s} and \eqref{eq:subspace-dim-s} that for $j\in \hat{\cR}\setminus (J\cup\{i\}),$ we have
	\begin{align*}
	\cA_1h_{j,1}^{(\hat{\cR})}=\cA_2h_{j,2}^{(\hat{\cR})}=\cdots=\cA_sh_{j,s}^{(\hat{\cR})}.
	\end{align*}
	Let us define 
	\begin{align}
		\cV_{\hat{\cR},i}^{(j)}=
		\begin{cases}
		\cA_1h_{j,1}^{(\hat{\cR})} & j\in \hat{\cR}\setminus (J\cup\{i\}),\\
		\cA_{\sigma(i)} & j\in J,
	\end{cases}
	\end{align}
	where $\sigma$ is a bijection between $J$ and $\{1,\ldots,s\}$ defined as before \eqref{eq:sigma}.
	
The procedure to construct a basis for optimal-access repair in this case is constructed as a modification of Algorithm~\ref{alg:basis-2}, and is
given in Algorithm~\ref{alg:basis-3}.
	
	\begin{algorithm}
%		\SetAlgoLined
		\KwIn{Subspaces $\cV_{\hat{\cR},i}^{(j)}$ for each $\hat{\cR}\subseteq \{1,\ldots,n\}$ such that $|\hat{\cR}|=d+1$ and $i\in\hat{\cR},j\in\hat{\cR}\setminus\{i\}$.}
		\KwOut{A basis $B$ for $K$ over $F$.}
		%initialization\;
		\ForEach{$\hat{\cR}\subseteq \{1,\ldots,n\}$ such that $|\hat{\cR}|=d+1$}{
			\ForEach{$i\in\hat{\cR}$}{
				\ForEach{$j\in\hat{\cR}\setminus\{i\}$}{
					$B_{\hat{\cR},i}^{(j)}\leftarrow \emptyset$\;
					$\cB_{\hat{\cR},i}^{(j)}\leftarrow \{0\}$\;
				}
			}
		}
		$\Omega\leftarrow \{(\hat{\cR},i,j)\mid  \hat{\cR}\subseteq \{1,\ldots,n\}, i\in\hat{\cR},j\in\hat{\cR}\setminus\{i\}\}$\;
		\For{$u\leftarrow 0$ \KwTo $\binom{n}{d+1}((d+1)^2-(d+1))-1$}{
			\ForEach{$I\subseteq \Omega$ such that $|I|=u$}{
				$\bar{I}\leftarrow \Omega\setminus I$\;
				$\cU_I\leftarrow\bigcap_{(\hat{\cR},i,j)\in\bar{I}}\cV_i^{(j)}$\;
				\ForEach{$\hat{\cR}\subseteq \{1,\ldots,n\}$ such that $|\hat{\cR}|=d+1$}{
					\ForEach{$i\in\hat{\cR}$}{
						\ForEach{$j\in\hat{\cR}\setminus\{i\}$}{
							\If{$(\hat{\cR},i,j)\in \bar{I}$}{
								$\cB_{\hat{\cR},i}^{(j)}\leftarrow \cB_{\hat{\cR},i}^{(j)}+\cU_I$\;
								Extend the set $B_{\hat{\cR},i}^{(j)}$ to be a basis of $\cB_{\hat{\cR},i}^{(j)}$ over $F$\;
							}
						}
					}
				}
			}
		}
		$\bar{B}\leftarrow \bigcup_{\hat{\cR}\subseteq \{1,\ldots,n\},|
		           \hat{\cR}|=d+1}\bigcup_{i\in\hat{\cR}}\bigcup_{j\neq \hat{\cR}\setminus\{i\}}B_{\hat{\cR},i}^{(j)}$\;
		Extend the set $\bar{B}$ to be a basis $B$ of $K$ over $F$\;
		\caption{Construction of an optimal basis; arbitrary repair degree}
		\label{alg:basis-3}
	\end{algorithm}
	
	Similarly to the previous sections, we have the following propositions, whose proofs are analogous to the proofs of Propositions~\ref{prop:bj} and
	\ref{prop:b}.
	
	\begin{proposition}
		\label{prop:bj-3}
When Algorithm~\ref{alg:basis-3} terminates, we have $\cB_{\hat{\cR},i}^{(j)}=\cV_{\hat{\cR},i}^{(j)}$ 
for $\hat{\cR}\subseteq\{1,\ldots,n\}$ with $|\hat{\cR}|=d+1$, $i\in\hat{\cR}$, and $j\in\hat{\cR}\setminus\{i\}$, and thus $B_{\hat{\cR},i}^{(j)}$ is a basis of $\cV_{\hat{\cR},i}^{(j)}$ over $F$.
	\end{proposition}
	
	\begin{proposition}
		\label{prop:b-3}
		Algorithm~\ref{alg:basis-3} returns a basis $B$ of $K$ over $F$.
	\end{proposition}

The basis of $K$ over $F$ constructed in the algorithm enables us to construct an optimal-access repair scheme for the code $\cC.$
Let $d\in\{k+1,\dots,n-1\}$ be the repair degree. Let $(c_1,\dots,c_n)$ be a codeword of the code $\cC$ written on the storage nodes,
and suppose that the failed node is $i$ and that $\cR$ be the set of $d$ helper nodes. 
Let $A$ be the block-diagonal matrix defined above, constructed with respect to $i$ and $H^{\hat{\cR}}$. 
Define the matrix $C^{\perp}=AH^{\hat\cR}$ and note that its rows $c_t^\bot, t=1,\dots,l$ form codewords of the code dual to the
punctured code $\cC_{\hat{\cR}}.$ Letting $c^{\perp}_{t}=(c^{\perp}_{t,i})_{i\in\cR},$ we can write
	\begin{align}
	c^{\perp}_{t,i}c_i = -\sum_{j\in\cR}c^{\perp}_{t,j}c_j.
	\end{align}
	Similarly to \eqref{eq:tr-constant-1}, we have
	\begin{align}
		\trace_{K/F}(c^{\perp}_{t,i}c_i)
		& = -\sum_{j\in\cR}\sum_{m=0}^{l-1}\trace_{K/F}(c^{\perp}_{t,j} b^*_m)c_{j,m},\label{eq:tr-constant-3}
	\end{align}
where $B^\ast=(b^\ast)$ is the dual basis of the basis $B.$ Note that for $j\in\cR$ we have
	\begin{align}
		\Span_{F}(c^{\perp}_{1,j},\ldots,c^{\perp}_{l,j}) = \Span_{F}(Ah_j) =\cV_{\cR,i}^{(j)}.
	\end{align}
By Algorithm~\ref{alg:basis-3} and Proposition~\ref{prop:bj-3}, the set $B_{\cR,i}^{(j)}\subseteq B$ forms a basis for 
the subspace $\cV_{\cR,j}^{(i)}$ over $F$. Therefore, the element $c^{\perp}_{t,j}$ can be linearly generated by the set 
$B_{\cR,i}^{(j)}$ for every $t=1,\ldots,l$. Let $B_{\cR,i}^{(j)}=\{b_{\cR,i,u}^{(j)}\,|\, u=1,\ldots,l/s \}$. Then, similarly to \eqref{eq:bj} and \eqref{eq:tr-constant-2}, we have
	\begin{align}
		c^{\perp}_{t,j} &= \sum_{u=1}^{l/s}\gamma_{j,u} b_{\cR,i,u}^{(j)},\\
		\trace_{K/F}(c^{\perp}_{t,i}c_i) &= -\sum_{j\in\cR}\sum_{m=0}^{l-1}\sum_{u=1}^{l/s}\trace_{K/F}(b_{\cR,i,u}^{(j)} b^{\ast}_m)\gamma_{j,u} c_{j,m}.
	\end{align}
	Therefore, each node $c_j,j\in\cR$ needs to access the set of symbol $\{c_{j,m}\,|\,\trace_{K/F}(b_{\cR,i,u}^{(j)} b^{\ast}_m) = 1 \}$, whose cardinality equals $|B_{\cR,i}^{(j)}|=l/s$. It follows that the constructed repair scheme has the optimal access property.
	
This completes the proof of Theorem \ref{thm:OA}.

\section{Concluding remarks}
We have shown that error correction is feasible in the original code family of \cite{Tamo18} without the increase of
the extension degree of the locator field of the code (the node size). Namely, codes from \cite{Tamo18} use extension degree $l=(d-k+1)L,$ where
$L$ is the product of the first $n$ distinct primes in an arithmetic progression, 
  $$
   L= \biggl(\prod_{\substack{i=1\\[.02in]p_i\equiv 1\text{ mod } (d-k+1)}}^n p_i\biggr).
  $$
The lower bound on $l$ from \cite{Tamo18}, necessary for repair of a single node, has the form $l \geq \prod_{i=1}^{k-1}p_i,$
where $p_i$ is the $i$-th smallest prime.  Asymptotically for fixed $d-k$ and growing $n$ we obtain the following bounds on the node size:
$\Omega(k^k)\le l\le O(n^n).$ Essentially the same node size is used in this paper for repair with error correction.
At the same time, the explicit RS code family with optimal access that we construct comes at the expense of larger node size, 
namely $l=(d-k+1)^nL.$ Since there is an optimal-access repair scheme for every scalar MSR code, this leaves a gap between what is known
explicitly and what is shown to be possible, which represents a remaining open question related to the task of optimal repair of RS codes.

\appendices
\section{Proof of Proposition \ref{prop:s}}\label{app:Prop4}

First we present the proof for the case $h=0$ (strictly speaking, we do not have to isolate it, but it makes understanding the general case much easier). In this case, definition \eqref{eq:fh}, \eqref{eq:imp} simplifies as follows. 
Let $f_0(x)=x^{p_1}-f(x).$ Write $f_0$ as
  \begin{align*}
   f_0(x)&=a_0+a_1 x +a_2 x^2 + \dots + a_{p_1-1}x^{p_1-1} \\
   &=\sum_{q=0}^{(p_1-1)/s-1} x^{qs} f_{0,q}(x),
   \end{align*}
where
  \begin{equation}\label{eq:ff}
 \left. \begin{aligned}
  f_{0,0}(x)&=a_0+a_1 x + \dots + a_{s-1}x^{s-1}\\
  f_{0,1}(x)&=a_s + a_{s+1}x + \dots + a_{2s-1}x^{s-1}\\
  &\dots\\
  f_{0,(p_1-1)/s-1}(x)&=a_{p_1-1-s}+a_{p_1-s}x+\dots+a_{p_1-1}x^s,
  \end{aligned}
  \right\}
  \end{equation}
so that the degree of the last polynomial is $\le s$ and the degrees of the remaining ones are $\le s-1.$
Obviously, we have
\begin{align}  \label{eq:rsn1}
\alpha_1^{p_1}&=f_0(\alpha_1)\\
& = \sum_{q=0}^{(p_1-1)/s-1} \alpha_1^{qs} f_{0,q}(\alpha_1).\label{eq:rsn2}
\end{align}

As before, we start with \eqref{eq:tin}, which implies that for any polynomial $g\in F_1[x]$ of degree $\deg g\le n-k-1$, we have
\begin{equation} \label{eq:t2}
\trace(e_iv_1 g(\alpha_1) c_1)=-\sum_{j=2}^n g(\alpha_j) \trace(e_iv_j c_j).
\end{equation}
Take $e_i=\alpha_1^{qs}$ and $g(x)=x^tf_{0,q}(\alpha_1)$ and sum on $q$ on the left, then from \eqref{eq:rsn2} we obtain $\trace(v_1\alpha_1^tf_0(\alpha_1)c_1)$.
Summing on $q$ on the right of \eqref{eq:t2} and using  \eqref{eq:rsn1}, we conclude that
   \begin{equation}\label{eq:t0}
   \trace(v_1 \alpha_1^{p_1+t} c_1)= - \sum_{q=0}^{(p_1-1)/s-1} \sum_{j=2}^n \alpha_j^t f_{0,q}(\alpha_j) \trace(\alpha_1^{qs} v_j c_j)
   \end{equation}
for all $t=0,1,\dots,n-k-s-1,$
Note that the constraint $t\le n-k-s-1$ is implied by the condition $\deg(g)=\deg(x^t f_{0,q}(x))\le n-k-1$ needed in order to use \eqref{eq:t2} (and \eqref{eq:tin}).
Change the variable $t\mapsto (t-1)$ to write the last equation as
\begin{equation} \label{eq:ss1}
\trace(v_1 \alpha_1^{p_1-1+t} c_1)
= - \sum_{q=0}^{(p_1-1)/s-1} \sum_{j=2}^n
\alpha_j^{t-1} f_{0,q}(\alpha_j) \trace(\alpha_1^{qs} v_j c_j). \quad\quad
t=1,2,\dots,n-k-s,
\end{equation}

From \eqref{eq:src} and the fact that
$$
\bigcap_{u=1}^{s-1} \{u-s,u-s+1,\dots,u-s+n-k-1\}
=\{-1, 0, 1, \dots, n-k-s\} ,
$$
we obtain
\begin{align*}
\trace(\beta^u \alpha_1^{p_1-1+t} v_1 c_1) 
= -\sum_{j=2}^n \alpha_j^{t-u+s} \trace(\beta^u \alpha_1^{u+p_1-s-1} v_j c_j), \quad  -1\le t\le n-k-s, \quad 1\le u\le s-1 .
\end{align*}
Summing these equations on $u=1,2,\dots,s-1,$ we obtain the relation
\begin{align*}
\trace\Big(\sum_{u=1}^{s-1} \beta^u \alpha_1^{p_1-1+t} v_1 c_1\Big) 
= -\sum_{j=2}^n \sum_{u=1}^{s-1} \alpha_j^{t-u+s} \trace(\beta^u \alpha_1^{u+p_1-s-1} v_j c_j), \quad  -1\le t\le n-k-s .
\end{align*}
For each $t=1,2,\dots,n-k-s$ let us add this equation and \eqref{eq:ss1}. This gives $n-k-s$ relations of the form
       \begin{multline*}
\trace\Big(\sum_{u=0}^{s-1} \beta^u \alpha_1^{p_1-1+t} v_1 c_1\Big) 
= -\sum_{j=2}^n \Big( \sum_{u=1}^{s-1} \alpha_j^{t-u+s} \trace(\beta^u \alpha_1^{u+p_1-s-1} v_j c_j)\\ +
\sum_{q=0}^{(p_1-1)/s-1} 
\alpha_j^{t-1} f_{0,q}(\alpha_j) \trace(\alpha_1^{qs} v_j c_j)\Big).
       \end{multline*}
Observe that the left-hand side of this equation is the same as the left-hand side of \eqref{eq:slp}. Therefore,
      \begin{align*}
\sum_{j=2}^n \alpha_j^t \trace\Big(\sum_{u=0}^{s-1}\beta^u \alpha_1^{p_1-1} v_j c_j\Big)
   = \sum_{j=2}^n &\Big( \sum_{u=1}^{s-1} \alpha_j^{t-u+s} \trace(\beta^u \alpha_1^{u+p_1-s-1} v_j c_j)\\ 
   +&
\sum_{q=0}^{(p_1-1)/s-1} 
\alpha_j^{t-1} f_{0,q}(\alpha_j) \trace(\alpha_1^{qs} v_j c_j)\Big), \quad 1\le t\le n-k-s .
      \end{align*}
Replacing $t-1$ with $t$ in this equation, we obtain that
      \begin{align*}
    \sum_{j=2}^n \alpha_j^t \Big( \sum_{u=1}^{s-1} \alpha_j^{s-u+1} \trace(\beta^u \alpha_1^{u+p_1-s-1} v_j &c_j) +
        \sum_{q=0}^{(p_1-1)/s-1} 
              f_{0,q}(\alpha_j) \trace(\alpha_1^{qs} v_j c_j)\\
           &- \alpha_j \trace\Big(\sum_{u=0}^{s-1}\beta^u \alpha_1^{p_1-1} v_j c_j\Big)\Big) = 0, \quad0\le t\le n-k-s-1 .
        \end{align*}
By Proposition \ref{prop:RS}, the vector
         \begin{equation} \label{eq:GRS3}
         \begin{aligned}
\Big(\sum_{u=1}^{s-1} \alpha_j^{s-u+1} \trace\Big(\beta^u \alpha_1^{u+p_1-s-1} v_j c_j\Big) +
&\sum_{q=0}^{(p_1-1)/s-1} 
 f_{0,q}(\alpha_j) \trace(\alpha_1^{qs} v_j c_j)\\
&- \alpha_j \trace\Big(\sum_{u=0}^{s-1}\beta^u \alpha_1^{p_1-1} v_j c_j\Big), j=2,\dots,n\Big)
      \end{aligned}
         \end{equation}
is contained in a GRS code of length $n-1$ and dimension $s+k-1$. This proves the case $h=0$ of the proposition.

Now let us consider the general case $0\le h\le s-1.$
From \eqref{eq:fh} and \eqref{eq:imp} we obtain
      \begin{equation}\label{eq:a1}
\alpha_1^{p_1+h}=f_h(\alpha_1) = \sum_{q=0}^{(p_1-1)/s-1} \alpha_1^{qs} f_{h,q}(\alpha_1).
     \end{equation}
This relation enables us to use the argument that yielded \eqref{eq:t0} above: Take $e_i=\alpha_1^{u+qs}\beta^u$ and $g(x)=x^t f_{h,q}(x)$ in
\eqref{eq:t2} and sum on $q=0,1,\dots,(p_1-1)/s-1$. We obtain for $h=0,\dots,s-1$ and $u=0,\dots,s-1-h$
    \begin{align*}
\trace(\alpha_1^{p_1+h+u+t} \beta^u v_1 c_1)
    &= \sum_{q=0}^{(p_1-1)/s-1} \trace(\alpha_1^{qs+u+t} \beta^u f_{h,q}(\alpha_1) v_1 c_1)\\
    &= - \sum_{q=0}^{(p_1-1)/s-1} \sum_{j=2}^n\alpha_j^t f_{h,q}(\alpha_j) \trace(\alpha_1^{u+qs} \beta^u v_j c_j) , \\
    &\hspace*{.5in}t=0,1,\dots,n-k-s-1.
    \end{align*}
The restriction $t\le n-k-s-1$ is imposed in the same way as in \eqref{eq:t0} (namely, it is necessary that $\deg(x^t f_{h,q}(x))\le n-k-1$). 
Replacing $h+u$ with $h$ in the last equation, we obtain that
\begin{align*}
\trace(\alpha_1^{p_1+h+t} \beta^u v_1 c_1)
= - \sum_{q=0}^{(p_1-1)/s-1} & \sum_{j=2}^n
\alpha_j^t f_{h-u,q}(\alpha_j) \trace(\alpha_1^{u+qs} \beta^u v_j c_j) , \\
& 0\le h\le s-1, \quad 0\le u\le h, \quad0\le t\le n-k-s-1.
\end{align*}
Let us sum these equations on $u=0,1,\dots,h$ to obtain
     \begin{align*}
\trace(\alpha_1^{p_1+h+t} \sum_{u=0}^h \beta^u v_1 c_1)
= - \sum_{j=2}^n \sum_{q=0}^{(p_1-1)/s-1} & \sum_{u=0}^h
\alpha_j^t f_{h-u,q}(\alpha_j) \trace(\alpha_1^{u+qs} \beta^u v_j c_j) , \\
& 0\le h\le s-1, \;0\le t\le n-k-s-1.
     \end{align*}
Replacing $t$ with $t-1$, we obtain that 
      \begin{equation} \label{eq:lrp}
     \begin{aligned}
    \trace\Big(\alpha_1^{p_1-1+h+t} \sum_{u=0}^h \beta^u v_1 c_1\Big)
         = - \sum_{j=2}^n \sum_{q=0}^{(p_1-1)/s-1} & \sum_{u=0}^h
       \alpha_j^{t-1} f_{h-u,q}(\alpha_j) \trace(\alpha_1^{u+qs} \beta^u v_j c_j) , \\
            & 0\le h\le s-1,\;1\le t\le n-k-s.
      \end{aligned}
      \end{equation}
According to \eqref{eq:src} and the fact that
$$
\bigcap_{u=h+1}^{s-1} \{u-s,u-s+1,\dots,u-s+n-k-1\}
=\{-1, 0, 1, \dots, n-k-s+h\} ,
$$
for $0\le h\le s-1$, we have
\begin{align*}
\trace(\beta^u \alpha_1^{p_1-1+t} v_1 c_1) 
= -\sum_{j=2}^n \alpha_j^{t-u+s} \trace(\beta^u \alpha_1^{u+p_1-s-1} v_j c_j), \\  -1\le t\le n-k-s+h, \; h+1\le u\le s-1 .
\end{align*}
Replacing $t$ with $t+h$, we have
\begin{align*}
\trace(\beta^u \alpha_1^{p_1-1+h+t} v_1 c_1) 
= -\sum_{j=2}^n \alpha_j^{h+t-u+s} \trace(\beta^u \alpha_1^{u+p_1-s-1} v_j c_j), \\  -h-1\le t\le n-k-s, \; h+1\le u\le s-1 .
\end{align*}
Summing these equations on $u=h+1,h+2,\dots,s-1$, we obtain
      \begin{align*}
\trace\Big(\sum_{u=h+1}^{s-1} \beta^u \alpha_1^{p_1-1+h+t} v_1 c_1\Big) 
= -\sum_{j=2}^n \sum_{u=h+1}^{s-1} \alpha_j^{h+t-u+s} \trace(\beta^u \alpha_1^{u+p_1-s-1} v_j c_j), \\  -h-1\le t\le n-k-s .
       \end{align*}
Finally, adding together this equation and \eqref{eq:lrp}, we obtain that
     \begin{align}
\trace\Big( \sum_{u=0}^{s-1} \beta^u \alpha_1^{p_1-1+h+t} v_1 c_1\Big)
= & - \sum_{j=2}^n \sum_{q=0}^{(p_1-1)/s-1}  \sum_{u=0}^h
\alpha_j^{t-1} f_{h-u,q}(\alpha_j) \trace(\alpha_1^{u+qs} \beta^u v_j c_j) \nonumber \\
& -\sum_{j=2}^n \sum_{u=h+1}^{s-1} \alpha_j^{h+t-u+s} \trace(\beta^u \alpha_1^{u+p_1-s-1} v_j c_j) , \label{eq:l1}\\
& \hspace*{1.2in} 0\le h\le s-1,\;
1\le t\le n-k-s. \nonumber
      \end{align}

Going back to \eqref{eq:slp}, let us perform the change $t\mapsto t+h,$ then we obtain
     \begin{equation}\label{eq:l2}
\trace\Big(\sum_{u=0}^{s-1}\beta^u \alpha_1^{p_1-1+h+t} v_1 c_1\Big) 
= -\sum_{j=2}^n \alpha_j^{h+t} \trace\Big(\sum_{u=0}^{s-1}\beta^u \alpha_1^{p_1-1} v_j c_j\Big), \quad\quad  -h\le t\le n-k-h-1 .
     \end{equation}
For $t=1,2,\dots,n-k-s$ the left-hand sides of \eqref{eq:l1} and \eqref{eq:l2} coincide, and therefore,
\begin{align*}
\sum_{j=2}^n \alpha_j^{h+t} \trace\Big(\sum_{u=0}^{s-1}\beta^u \alpha_1^{p_1-1} v_j c_j\Big)
= &  \sum_{j=2}^n \sum_{q=0}^{(p_1-1)/s-1}  \sum_{u=0}^h
\alpha_j^{t-1} f_{h-u,q}(\alpha_j) \trace(\alpha_1^{u+qs} \beta^u v_j c_j)  \\
& + \sum_{j=2}^n \sum_{u=h+1}^{s-1} \alpha_j^{h+t-u+s} \trace(\beta^u \alpha_1^{u+p_1-s-1} v_j c_j) ,  \quad\quad
1\le t\le n-k-s.
\end{align*}
Replacing $t$ by $t+1$, we obtain that
\begin{align*}
\sum_{j=2}^n \alpha_j^t \Big( \sum_{q=0}^{(p_1-1)/s-1}  \sum_{u=0}^h
 f_{h-u,q}(\alpha_j) \trace(\alpha_1^{u+qs} \beta^u v_j c_j)  
 +  \sum_{u=h+1}^{s-1} \alpha_j^{h+1-u+s} \trace(\beta^u \alpha_1^{u+p_1-s-1} v_j c_j) \\
- \alpha_j^{h+1} \trace\Big(\sum_{u=0}^{s-1}\beta^u \alpha_1^{p_1-1} v_j c_j\Big) \Big) = 0, \\
0\le t\le n-k-s-1.
\end{align*}
The proof is complete.

	\bibliographystyle{IEEEtran}
	\bibliography{./RS-io}
\end{document}